\documentclass[10pt, doublecolumn]{IEEEtran}
\usepackage[dvips]{graphicx}
\usepackage{graphics, subfigure, color, epsfig, psfrag, amsmath, amsfonts, amssymb, amsthm, eucal, balance}
\usepackage{bm}

\newtheorem{cor}{Corollary}
\newtheorem{lem}{Lemma}
\newtheorem{prop}{Proposition}

\begin{document}
\title{Mixed-ADC Massive MIMO}
\author{Ning Liang, Wenyi Zhang, \emph{Senior Member, IEEE}
\thanks{
Manuscript received April 14, 2015; revised September 8, 2015; accepted December 11, 2015. Review of this manuscript was coordinated by Stefano Buzzi, lead guest editor of this issue. Part of this work has been presented at IEEE Information Theory Workshop, Jeju Island, Korea, Oct. 2015.

N. Liang and W. Zhang are with Key Laboratory of Wireless-Optical Communications, Chinese Academy of Sciences, Department of Electronic Engineering and Information Science, University of Science and Technology of China, Hefei, China (Emails: liangn@mail.ustc.edu.cn, wenyizha@ustc.edu.cn).

This work has been supported by the National Basic Research Program of China (973 Program) through grant 2012CB316004, and National Natural Science Foundation of China through grant 61379003.
}
}
\maketitle
\thispagestyle{empty}
\begin{abstract}
Motivated by the demand for energy-efficient communication solutions in the next generation cellular network, a mixed-ADC architecture for massive multiple input multiple output (MIMO) systems is proposed, which differs from previous works in that herein one-bit analog-to-digital converters (ADCs) partially replace the conventionally assumed high-resolution ADCs. The information-theoretic tool of generalized mutual information (GMI) is exploited to analyze the achievable data rates of the proposed system architecture and an array of analytical results of engineering interest are obtained. For fixed single input multiple output (SIMO) channels, a closed-form expression of the GMI is derived, based on which the linear combiner is optimized. The analysis is then extended to ergodic fading channels, for which tight lower and upper bounds of the GMI are obtained. Impacts of dithering and imperfect channel state information (CSI) are also investigated, and it is shown that dithering can remarkably improve the system performance while imperfect CSI only introduces a marginal rate loss. Finally, the analytical framework is applied to the multi-user access scenario. Numerical results demonstrate that the mixed-ADC architecture with a relatively small number of high-resolution ADCs is able to achieve a large fraction of the channel capacity of conventional architecture, while reduce the energy consumption considerably even compared with antenna selection, for both single-user and multi-user scenarios.
\end{abstract}
\setcounter{page}{1}

\section{Introduction}
The exponential increase in the demand for mobile data traffic imposes great challenge on the cellular network. In recent years, a heightened attention has been focused on massive multiple input multiple output (MIMO) systems, in which each base station (BS) is equipped with hundreds of antennas and serves tens of or more users simultaneously \cite{marzetta2010noncooperative}-\cite{rusek2013scaling}. Because the large number of BS antennas can effectively average out noise, fading and to some extent, noncoherent interference, massive MIMO achieves significant gains in both spectral efficiency and radiated energy efficiency, and thus is envisioned as a promising key enabler for the next generation cellular network \cite{ngo2013energy}-\cite{boccardi2014five}.

Thus far, most of the literature on massive MIMO assume a conventional architecture built on ideal hardware. However, this assumption is not well justified, since the hardware cost and circuit power consumption scale linearly with the number of BS antennas and thus soon become practically unbearable unless low-cost, energy-efficient hardware is deployed which however easily suffers from impairments. Assuming an additive stochastic impairment model, the authors of \cite{bjornson2014hardware} examined the impact of hardware impairments on both spectral efficiency and radiated energy efficiency of massive MIMO. The authors of \cite{bjornson2015massive} obtained scaling law that describes how fast the tolerance level of impairments increases with the number of BS antennas while reaping much of the performance gain promised by massive MIMO. The authors of \cite{gustavsson2014impact} examined the accuracy of widely used additive or multiplicative stochastic impairment models by providing a hardware-specific deterministic model and performing comparative numerical studies.

Due to the favorable property of low cost, low power consumption and feasibility of implementation \cite{walden1999analog}-\cite{le2005analog}, low-resolution analog-to-digital converters (ADCs) have also attracted ubiquitous attention in the field of energy-efficient design for wireless communication systems. For Nyquist-sampled real Gaussian channel, the authors of \cite{singh2009limits} established some general results regarding low-resolution quantization, showing that for a quantizer with $Q$ bins, the capacity-achieving input alphabet should be discrete and needs not have more than $Q$ mass points. The authors of \cite{mezghani2007modified} designed a modified minimum mean square error (MMSE) receiver for MIMO systems with output quantization and proposed a lower bound to the capacity. In \cite{yin2010monobit}, the authors investigated a practical monobit digital receiver paradigm for impulse radio ultra-wideband (UWB) systems. Recently, the authors of \cite{risi2014massive} examined the impact of one-bit quantization on achievable rates of massive MIMO systems with both perfect and estimated channel state information (CSI). The authors of \cite{mo2014high} addressed the high signal-to-noise ratio (SNR) capacities of both single input multiple output (SIMO) and MIMO channels with one-bit output quantization.

Despite its great superiority in deployment cost and energy efficiency, one-bit quantization generally has to tolerate large rate loss, especially in the high SNR regime \cite{mo2014high}, thus highlighting the indispensability of high-resolution ADC for digital receiver. Besides, the great overhead of pilot-aided channel estimation under one-bit quantization is also a big concern \cite{yin2010monobit}-\cite{risi2014massive}, \cite{ivrlac2007mimo}. Thus motivated by such consideration, in this paper we propose a mixed-ADC architecture for massive MIMO systems in which one-bit ADCs partially, but not completely, replace conventionally assumed high-resolution ADCs. This architecture has the potential of allowing us to remarkably reduce the hardware cost and power consumption while still maintain a large fraction of the performance gains promised by conventional architecture.

For such mixed-ADC massive MIMO, although the channel capacity is still the maximum mutual information between the channel input and the quantized channel output vector, from an engineering perspective, however, the mutual information maximization problem appears to be not completely satisfactory in providing engineering insights. Because in this situation, the mutual information is high-dimensional integration and summation which do not yield closed-form simplification as in linear Gaussian channels. Generalized mutual information (GMI) \cite{ganti2000mismatched}-\cite{lapidoth2002fading}, on the other hand, allows one to analytically characterize the achievable date rates of low-complexity linear receivers that are particularly favorable for massive MIMO systems, and thus we leverage it to address the performance of the mixed-ADC architecture. As a performance metric for mismatched decoding, GMI has proved convenient and useful in several important scenarios such as fading channels with imperfect CSI at the receiver \cite{lapidoth2002fading}, channels with transceiver distortion \cite{zhang2012general}-\cite{vehkapera2015asymptotic} and analysis of bit-interleaved coded modulation \cite{guillen2008bit}.

Exploiting a general analytical framework developed in \cite{zhang2012general}, we obtain a series of analytical results. First, we consider a fixed SIMO channel where the BS is equipped with $N$ antennas but only has access to $K$ pairs\footnote{A pair of ADCs quantize the I/Q components of an antenna, respectively.} of high-resolution ADCs and $(N-K)$ pairs of one-bit ADCs, and derive a closed-form expression of the GMI. This enables us to optimize the linear combiner and further explore the asymptotic behaviors of the GMI in both low and high SNR regimes that in turn suggest a plausible ADC switch scheme. Besides, the benefit of dithering is also investigated, for which we propose a simple but effective dithering scheme, which achieves remarkable rate gain, especially for the case of small $K$.

The analysis is then extended to the scenario of ergodic fading channels where, instead of directly working with the exact GMI, we derive lower and upper bounds of the GMI, which are shown to be very tight by numerical study. Moreover, numerical results reveal that the mixed-ADC architecture with a small number of high-resolution ADCs suffices to attain a large portion of the channel capacity of conventional architecture and meanwhile outperforms antenna selection with the same number of high-resolution ADCs\footnote{In conventional architecture, each BS antenna is followed by a radio frequency (RF) chain built on ideal hardware. Meanwhile, by antenna selection we mean that there are only $K$ ideal RF chains available at the BS.}. The robustness of the mixed-ADC architecture against imperfect CSI is also investigated. In this paper, we only utilize the high-resolution ADCs to perform channel estimation, and thus the deduced estimation error is Gaussian distributed in Rayleigh fading channels, allowing us to analytically characterize the resulting GMI as well as its lower and upper bounds. Numerical results show that the lower and upper bounds are again very tight and that there is only a marginal rate loss due to imperfect CSI.

Finally, we apply our analysis to the multi-user access scenario. The corresponding numerical results indicate that when equipped with a small number of high-resolution ADCs, the mixed-ADC architecture also achieves a large fraction of the achievable rate of conventional architecture and again outperforms antenna selection with the same number of high-resolution ADCs.

In addition, energy efficiencies of the mixed-ADC architecture and of antenna selection are compared, taking that of conventional architecture as a baseline. Numerical results reveal that under the same spectral efficiency loss, both the mixed-ADC architecture and antenna selection achieve significant energy reduction. Moreover, the mixed-ADC architecture always outperforms antenna selection, especially in the multi-user scenario. In summary, the mixed-ADC architecture strikes an attractive balance between spectral efficiency and energy efficiency, for both single-user and multi-user scenarios.

The remaining part of this paper is organized as follows. Section \ref{sect:system architecture} outlines the system model. Adopting GMI as the performance metric, Section \ref{sect:GMI and combiner} establishes the theoretical framework for fixed SIMO channels, based on which the optimal linear combiner and the asymptotic behaviors of the GMI in both low and high SNR regimes are explored. Besides, performance improvement through dithering is also investigated. Then, Section \ref{sect:ergodic fading channel} extends the theoretical framework to ergodic fading channels and evaluates the the effects of imperfect CSI on the system performance. Section \ref{sect:multi-user} applies the theoretical framework to the multi-user access scenario. Furthermore, energy efficiency of the mixed-ADC architecture is assessed in Section
\ref{sect:energy efficiency}. Numerical results are presented in Section \ref{sect:numerical} to corroborate the analysis. Finally, Section \ref{sect:conclusion} concludes the paper. Auxiliary technical derivations are archived in the appendix.

\emph{Notation:} Throughout this paper, vectors and matrices are given in bold typeface, e.g., $\mathbf{x}$ and $\mathbf{X}$, respectively, while scalars are given in regular typeface, e.g., $x$. We use $\|\mathbf{x}\|_{1}$ and $\|\mathbf{x}\|$ to represent the 1-norm and 2-norm of vector $\mathbf{x}$, respectively, and let $\mathbf{X}^{*}$, $\mathbf{X}^{T}$ and $\mathbf{X}^{H}$ denote the conjugate, transpose and conjugate transpose of $\mathbf{X}$, respectively. Normal distribution with mean $\mu$ and variance $\sigma^2$ is denoted by $\mathcal{N}(\mu,\sigma^2)$, while $\mathcal{CN}(\bm{\mu},\mathbf{C})$ stands for the distribution of a circularly symmetric complex Gaussian random vector with mean $\bm{{\mu}}$ and covariance matrix $\mathbf{C}$. Superscripts $\mathrm{R}$ and $\mathrm{I}$ are used to indicate the real and imaginary parts of a complex number, respectively, e.g., $x=x^{\mathrm{R}}+i\cdot x^{\mathrm{I}}$, with $i$ being the imaginary unit. We use $\mathrm{sgn}(x)=\mathrm{sgn}(x^{\mathrm{R}})+i\cdot\mathrm{sgn}(x^{\mathrm{I}})$ to denote the sign function of a complex number $x$, and $\log(x)$ to denote the natural logarithm of positive real number $x$.

\section{System model}
\label{sect:system architecture}
Several scenarios will be addressed in this paper, including fixed SIMO channels, ergodic fading SIMO channels with perfect or imperfect CSI at the receiver, and multi-user channels with multiple single-antenna users and a multi-antenna BS. In this section, we describe the fixed SIMO channel model, and the remaining scenarios will be introduced in later sections.
\begin{figure*}
\centering
\includegraphics[width=0.8\textwidth]{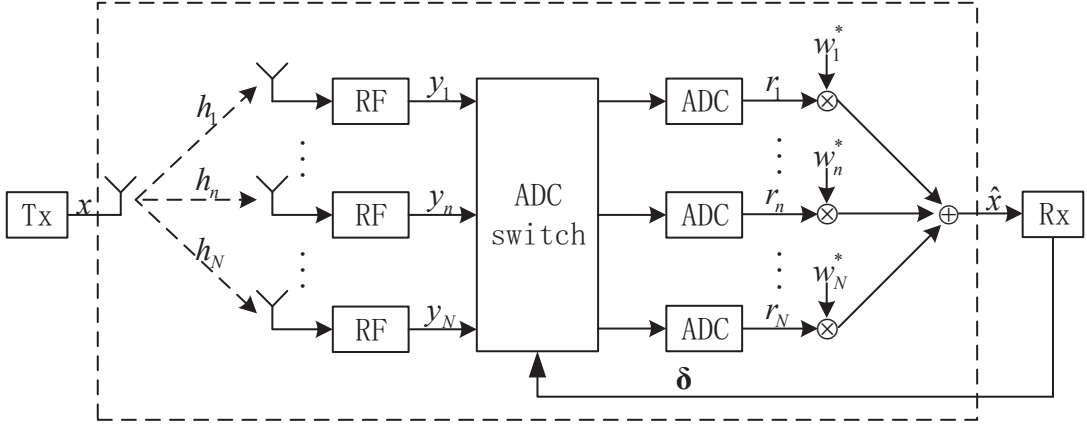}
\centering
\caption{Illustration of the system architecture. It is perhaps worth noting that the ADC switch module can also be placed before the RF chains. In this manner, the RF chain followed by a
pair of one-bit ADCs can be manufactured with lower quality requirements and consequently we can further reduce the power consumption and hardware cost. On the other hand, switch at radio frequency may be more challenging and costly than at baseband. Which choice is favorable will be determined by practical engineering.} \label{fig:SystemModel}
\end{figure*}

As aforementioned, we consider a single-user system, where a single-antenna user communicates with an $N$-antenna BS. Moreover, we consider a narrow-band channel model\footnote{Throughout this paper we focus on a narrow-band channel model, similar to those considered in, e.g., \cite{bjornson2014hardware}-\cite{gustavsson2014impact}, \cite{risi2014massive}-\cite{mo2014high}, \cite{vehkapera2015asymptotic}, among others. Wideband channel model includes multi-path effect, which can still be treated using the general framework of GMI, and will be treated in a separate work; a further discussion is in Section \ref{sect:conclusion}.}, for which the channel vector $\mathbf{h}$ is fixed throughout the transmission of the codeword and is assumed to be perfectly known by the BS. Then the received signal at the BS can be expressed as
\begin{equation}
\mathbf{y}^l=\mathbf{h}x^l+\mathbf{z}^l,\ \ \ \mathrm{for}\  l=1,2,...,L,
\label{equ:equ_1}
\end{equation}
where $x^l$ is the complex signal transmitted at the $l$-th symbol time, $\mathbf{z}^l\sim\mathcal{CN}(\mathbf{0},\sigma^2\mathbf{I})$ models the independent and identically distributed (i.i.d.) complex Gaussian noise vector, and $L$ is the codeword length.

In practice, the received signal at each antenna is quantized by a pair of ADCs, one for each of the in-phase and quadrature (I/Q) branches, so that further signal processing can be performed in the digital domain. Despite of this, most of the literature on receiver design assume ADC with virtually infinite precision for the tractability of analysis. For a large BS antenna array, however, such assumption is no longer justified since the cost and energy consumption of conventional architecture scale linearly with the number of BS antennas, which will soon become the system bottleneck. Therefore, we propose a mixed-ADC architecture in which only $2K$ high-resolution ADCs are available and all the other $2(N-K)$ ADCs are with only one-bit resolution\footnote{Note that a one-bit ADC is particularly simple to implement in hardware, say, using a polarity detector \cite{yin2010monobit}. Furthermore, the analytical approach we adopt in this work, based on the general framework in \cite{zhang2012general}, can be extended to other types of ADCs.}. We further let the I/Q outputs at each antenna be quantized by two ADCs of the same kind. Thus the quantized output is
\begin{equation}
r_{n}^{l}=\mathrm{Q}_{n}(y_n^l)=
\begin{cases}
h_{n} x^{l}+z_{n}^{l}, &\mathrm{if}\  \delta_n=1,\\
\mathrm{sgn}(h_{n} x^{l}+z_{n}^{l}), &\mathrm{if}\  \delta_n=0,
\end{cases}
\label{equ:equ_2}
\end{equation}
for $l=1,...,L,\ n=1,...,N$. Here $\delta_n\in\{0,1\}$ is an indicator: $\delta_n=1$ means that the ADCs corresponding to the $n$-th antenna are high-resolution, whereas $\delta_n=0$ indicates that they are with one-bit resolution. Here for simplicity we assume sufficiently high resolution for $\delta_n=1$, so that the residual quantization noise is negligible then.

To make the expression compact, we introduce $\bar{\delta}_n\triangleq 1-\delta_n$ and rewrite \eqref{equ:equ_2} as
\begin{equation}
r_{n}^{l}=\delta_{n}\cdot(h_{n} x^{l}+z_{n}^{l})+\bar{\delta}_n\cdot\mathrm{sgn}(h_{n} x^{l}+z_{n}^{l}).
\label{equ:equ_3}
\end{equation}
Then, we define an ADC switch vector $\bm{\delta}\triangleq[\delta_1,...,\delta_N]^{T}$, which follows the subsequent restriction
\begin{equation}
\|\bm{\delta}\|_1=\sum_{n=1}^{N}\delta_n=K,
\label{equ:equ_4}
\end{equation}
and should be optimized according to the channel $\mathbf{h}$ so that the limited number of high-resolution ADCs will be well utilized to enhance the system performance.

For transmission of rate $R$, the user selects a message $m$ from $\mathcal{M}=\{1,2,...,\lfloor 2^{LR} \rfloor\}$ uniformly randomly, and maps the selected message to a transmitted codeword, i.e., a length-$L$ complex sequence, $\{x^l(m)\}_{l=1}^{L}$. In this paper, we restrict the codebook to be drawn from a Gaussian ensemble; that is, each codeword is a sequence of $L$ i.i.d. $\mathcal{CN}(0,\mathcal{E}_\mathrm{s})$ random variables, and all the codewords are mutually independent. Such a choice of codebook satisfies the average power constraint $\frac{1}{L}\sum_{l=1}^{L}\mathbb{E}[|x^{l}(m)|^2]\leq\mathcal{E}_\mathrm{s}$. We define the SNR as $\mathrm{SNR}=\mathcal{E}_\mathrm{s}/\sigma^2$, and let $\sigma^2=1$ thereafter for convenience.

As is well known, without receiver distortion, the Gaussian codebook ensemble together with nearest-neighbor decoding achieves the capacity of conventional architecture\footnote{For an $N$-antenna SIMO channel, we let $C(N,K_1,K_2)$ denote its capacity when equipped with $K_1$ pairs of high-resolution ADCs and $K_2$ pairs of one-bit ADCs, where $0\leq K_1, K_2, K_1+K_2 \leq N$. Particularly, for the mixed-ADC architecture, we have $K_1=K$ and $K_2=N-K$; for antenna selection, we have $K_1=K$ and $K_2=0$, discarding the outputs of $(N-K)$ antennas.} $C(N,N,0)=\mathrm{log}(1+\|\mathbf{h}\|^2\mathrm{SNR})$, as the codeword length $L$ grows without bound. With $(N - K)$ pairs of one-bit ADCs, the channel capacity $C(N,K,N-K)$ is less than $C(N,N,0)$ due to information loss during quantization.

As discussed in the introduction, instead of numerically evaluating $C(N,K,N-K)$, in the following, we adopt the nearest-neighbor decoding rule at the decoder, and leverage the general framework developed in \cite{zhang2012general} to investigate the GMI of the mixed-ADC architecture. The GMI acts as an achievable rate and thus also a lower bound of $C(N,K,N-K)$. To this end, we introduce a linear combiner\footnote{There should be some nonlinear receiver that outperforms the linear one in this paper, which will be studied in a future work.} to process the channel output vector, as illustrated in Figure \ref{fig:SystemModel}. Thus the processed channel output is
\begin{equation}
\hat{x}^{l}=\mathbf{w}^{H}\mathbf{r}^{l},
\label{equ:equ_6}
\end{equation}
for $l=1,...,L$, where $\mathbf{w}$ is designed according to the channel $\mathbf{h}$ and the ADC switch vector $\bm{\delta}$.

With nearest-neighbor decoding, upon observing $\{\hat{x}^{l}\}_{l=1}^L$, the decoder computes, for all messages, the Euclidean distances
\begin{equation}
D(m)=\frac{1}{L}\sum_{l=1}^{L}|\hat{x}^{l}-a x^{l}(m)|^2,\ \ \ m\in\mathcal{M},
\label{equ:equ_7}
\end{equation}
and decides the received message as the one that minimizes \eqref{equ:equ_7}. Here the scaling parameter $a$ is adopted to adjust the power imbalance between the channel input $x^l$ and the processed output $\hat{x}^l$ contributed collectively by the channel, one-bit quantization and the linear combiner, and should be selected appropriately for optimizing the decoding performance.

\section{GMI and Optimal Combining}
\label{sect:GMI and combiner}
\subsection{GMI of the Proposed System Framework}
\label{subsect:GMI}
From now on, we suppress the time index $l$ for notational simplicity. To facilitate the exposition, we summarize \eqref{equ:equ_3} and \eqref{equ:equ_6} as
\begin{equation}
\hat{x}=\mathbf{w}^{H}\mathbf{r}\triangleq f(x,\mathbf{h},\mathbf{z}),
\label{equ:equ_8}
\end{equation}
where $f(\cdot)$ is a memoryless nonlinear distortion function that incorporates the effects of output quantization as well as linear combining, and maps the triple $(x,\mathbf{h},\mathbf{z})$ into the processed output $\hat{x}$. Although $\bm{\delta}$ and  $\mathbf{w}$ are made invisible in the function $f(\cdot)$ since they are both determined by $\mathbf{h}$, we need to keep in mind that $f(\cdot)$ implicitly includes $\bm{\delta}$ and $\mathbf{w}$.

We apply the general framework developed in \cite{zhang2012general} to derive the GMI of the system architecture. The GMI is a lower bound of the channel capacity, and more precisely, it characterizes the maximum achievable rate under the specified random codebook (Gaussian ensemble here) and the specified decoding rule (nearest-neighbor decoding here) such that the average decoding error probability (averaged over the codebook ensemble) is guaranteed to vanish asymptotically as the codeword length grows without bound \cite{lapidoth2002fading}. Particularly, conditioned on $\mathbf{w}$ and $\bm{\delta}$, the GMI takes the following form analogous to \cite[Eq. (89)]{zhang2012general}; that is,
\begin{align}
I_{\mathrm{GMI}}(\mathbf{w},\bm{\delta})=&\sup_{a\in\mathbb{C},\theta<0}
\Bigg(\theta\mathbb{E}[|f(x,\mathbf{h},\mathbf{z})-ax|^2]-\nonumber\\
&\frac{\theta\mathbb{E}[|f(x,\mathbf{h},\mathbf{z})|^2]}{1-\theta|a|^2\mathcal{E}_\mathrm{s}}
\!+\!\log(1\!-\!\theta|a|^2\mathcal{E}_\mathrm{s})\Bigg),
\label{equ:equ_49}
\end{align}
where the expectation is taken with respect to $x$ and $\mathbf{z}$. The parameter $a$ is in the nearest-neighbor decoding rule (\ref{equ:equ_7}), and the parameter $\theta$ is from the underlying large-deviations argument, --- for further details about the derivation of the expression, we refer to \cite{lapidoth2002fading} \cite{zhang2012general}. Then we can solve the optimization problem in (\ref{equ:equ_49}), following essentially the same line as \cite[App. C]{zhang2012general}, and obtain an explicit expression of the GMI as follows.

\begin{prop}
\label{prop:prop_1}
With Gaussian codebook ensemble and nearest-neighbor decoding, the GMI for given $\mathbf{w}$ and $\bm{\delta}$ is
\begin{equation}
I_{\mathrm{GMI}}(\mathbf{w},\bm{\delta})=\mathrm{log}\left(1+\frac{\kappa(\mathbf{w},\bm{\delta})}{1-\kappa(\mathbf{w},\bm{\delta})}\right),
\label{equ:equ_9}
\end{equation}
where the parameter $\kappa(\mathbf{w},\bm{\delta})$ is
\begin{equation}
\kappa(\mathbf{w},\bm{\delta})=\frac{|\mathbb{E}[f^{*}(x,\mathbf{h},\mathbf{z})\cdot x]|^2}
{\mathcal{E}_\mathrm{s} \mathbb{E}[|f(x,\mathbf{h},\mathbf{z})|^2]}.
\label{equ:equ_10}
\end{equation}
The corresponding optimal choice of the scaling parameter $a$ is
\begin{equation}
a_{\mathrm{opt}}(\mathbf{w},\bm{\delta})=\frac{\mathbb{E}[f(x,\mathbf{h},\mathbf{z})\cdot x^{*}]}{\mathcal{E}_\mathrm{s}}.
\label{equ:equ_11}
\end{equation}
We note that the expectation is taken with respect to $x$ and $\mathbf{z}$.
\end{prop}

It is worth noting that $\kappa(\mathbf{w},\bm{\delta})$ is the squared correlation coefficient of channel input $x$ and the processed output $f(x,\mathbf{h},\mathbf{z})$, and thus is upper bounded by one, from Cauchy-Schwartz's inequality. Moreover, $I_{\mathrm{GMI}}(\mathbf{w},\bm{\delta})$ is a strictly increasing function of $\kappa(\mathbf{w},\bm{\delta})$ for $\kappa(\mathbf{w},\bm{\delta})\in(0,1)$. Therefore, in the following, we will seek to maximize $\kappa(\mathbf{w},\bm{\delta})$ by choosing well designed linear combiner $\mathbf{w}$ and ADC switch vector $\bm{\delta}$. To this end, we first derive a closed-form expression for $\kappa(\mathbf{w},\bm{\delta})$. The result is summarized by the following proposition.

\begin{prop}
\label{prop:prop_2}
Given $\mathbf{w}$ and $\bm{\delta}$, for \eqref{equ:equ_10} in Proposition \ref{prop:prop_1}, we have
\begin{equation}
\kappa(\mathbf{w},\bm{\delta})=\frac{\mathbf{w}^{H}\mathbf{R}_{\mathbf{r}x}\mathbf{R}_{\mathbf{r}x}^{H}\mathbf{w}}
{\mathcal{E}_\mathrm{s}\mathbf{w}^{H}\mathbf{R}_{\mathbf{rr}}\mathbf{w}},
\label{equ:equ_17}
\end{equation}
where $\mathbf{R}_{\mathbf{r}x}$ is the correlation vector between $\mathbf{r}$ and $x$, with its $n$-th element being
\begin{equation}
(\mathbf{R}_{\mathbf{r}x})_n=h_n \mathcal{E}_\mathrm{s} \left[\delta_n+\bar{\delta}_n\cdot\sqrt{\frac{4}{\pi(|h_n|^2\mathcal{E}_\mathrm{s}+1)}}\right],
\label{equ:equ_18}
\end{equation}
and $\mathbf{R}_{\mathbf{rr}}$ is the covariance matrix of $\mathbf{r}$, with its $(n,m)$-th entry being $(\mathbf{R}_{\mathbf{rr}})_{n,m}=$
\begin{equation}
\begin{cases}
1+\delta_n\cdot |h_n|^2\mathcal{E}_\mathrm{s}+\bar{\delta}_n,&\mathrm{if}\ n=m, \\
h_n h_m^*\mathcal{E}_\mathrm{s}\Bigg[\delta_n\delta_m+
\delta_n\bar{\delta}_m\cdot\sqrt{\frac{4}{\pi(|h_m|^2\mathcal{E}_\mathrm{s}+1)}}+\\
\ \ \ \ \ \ \ \ \ \ \bar{\delta}_n\delta_m\cdot\sqrt{\frac{4}{\pi(|h_n|^2\mathcal{E}_\mathrm{s}+1)}}\Bigg]+\\
\bar{\delta}_n\bar{\delta}_m\!\cdot\!\frac{4}{\pi}\Bigg[
\mathrm{arcsin}\Big(\frac{(h_nh_m^*)^{\mathrm{R}}\mathcal{E}_\mathrm{s}}
{\sqrt{|h_n|^2\mathcal{E}_\mathrm{s}+1}\sqrt{|h_m|^2\mathcal{E}_\mathrm{s}+1}}\Big)+\\
\ \ \ \ \ \ \ \ \ \ i\!\cdot\!\mathrm{arcsin}\Big(\frac{(h_nh_m^*)^{\mathrm{I}}\mathcal{E}_\mathrm{s}}
{\sqrt{|h_n|^2\mathcal{E}_\mathrm{s}+1}\sqrt{|h_m|^2\mathcal{E}_\mathrm{s}+1}}\Big)
\Bigg],&\mathrm{if}\ n\neq m.
\end{cases}
\label{equ:equ_19}
\end{equation}
The corresponding optimal choice of the scaling parameter $a$ in \eqref{equ:equ_11} is
\begin{equation}
a_{\mathrm{opt}}(\mathbf{w},\bm{\delta})=\frac{1}{\mathcal{E}_\mathrm{s}}\mathbf{w}^{\mathrm{H}}\mathbf{R}_{\mathbf{r}x}.
\label{equ:equ_20}
\end{equation}
\end{prop}

\begin{proof}
See Appendix-A.
\end{proof}

\subsection{Optimization of Linear Combiner}
\label{subsect:combiner}
In the previous subsection, the GMI of the system architecture is derived, as a function of $\mathbf{h}$, $\mathbf{w}$ and $\bm{\delta}$. In this subsection, we turn to the optimization of $\mathbf{w}$ such that the GMI is maximized for given $\mathbf{h}$ and $\bm{\delta}$. The subsequent proposition summarizes our result.

\begin{prop}
\label{prop:prop_3}
For given $\mathbf{h}$ and $\bm{\delta}$, the optimal linear combiner $\mathbf{w}$ takes the following form
\begin{equation}
\mathbf{w}_{\mathrm{opt}}=\mathbf{R}_{\mathbf{rr}}^{-1}\mathbf{R}_{\mathbf{r}x},
\label{equ:equ_21}
\end{equation}
which is in fact a linear MMSE combiner that minimizes the mean squared estimation error of $x$ upon observing $\mathbf{r}$ among all linear combiners. The corresponding $\kappa(\mathbf{w},\bm{\delta})$ is
\begin{equation}
\kappa(\mathbf{w}_{\mathrm{opt}},\bm{\delta})=a_{\mathrm{opt}}(\mathbf{w}_{\mathrm{opt}},\bm{\delta})=
\frac{1}{\mathcal{E}_\mathrm{s}}\mathbf{R}_{\mathbf{r}x}^{H}\mathbf{R}_{\mathbf{rr}}^{-1}\mathbf{R}_{\mathbf{r}x}.
\label{equ:equ_22}
\end{equation}
\end{prop}

\begin{proof}
Noticing that $\mathbf{R}_{\mathbf{rr}}$ is a positive semidefinite Hermitian matrix, from \eqref{equ:equ_17} we have
\begin{eqnarray}
\kappa(\mathbf{w},\bm{\delta})&=&\frac{1}{\mathcal{E}_\mathrm{s}}
\frac{|\mathbf{w}^{H}\mathbf{R}_{\mathbf{r}x}|^2}{\mathbf{w}^{H}\mathbf{R}_{\mathbf{rr}}\mathbf{w}}\nonumber\\
&=&\frac{1}{\mathcal{E}_\mathrm{s}}
\frac{|\mathbf{w}^{H}\mathbf{R}_{\mathbf{rr}}^{1/2}\mathbf{R}_{\mathbf{rr}}^{-1/2}\mathbf{R}_{\mathbf{r}x}|^2}
{\|\mathbf{w}^{H}\mathbf{R}_{\mathbf{rr}}^{1/2}\|^2}\nonumber\\
&\leq&\frac{1}{\mathcal{E}_\mathrm{s}}
\frac{\|\mathbf{w}^{H}\mathbf{R}_{\mathbf{rr}}^{1/2}\|^2\cdot\|\mathbf{R}_{\mathbf{rr}}^{-1/2}\mathbf{R}_{\mathbf{r}x}\|^2}
{\|\mathbf{w}^{H}\mathbf{R}_{\mathbf{rr}}^{1/2}\|^2}\nonumber\\
&=&\frac{1}{\mathcal{E}_\mathrm{s}}\|\mathbf{R}_{\mathbf{rr}}^{-1/2}\mathbf{R}_{\mathbf{r}x}\|^2,
\label{equ:equ_23}
\end{eqnarray}
where the inequality follows from Cauchy-Schwartz's inequality, which holds equality if and only if $\mathbf{w}^{H}\mathbf{R}_{\mathbf{rr}}^{1/2}=(\mathbf{R}_{\mathbf{rr}}^{-1/2}\mathbf{R}_{\mathbf{r}x})^{H}$, i.e., $\mathbf{w}_{\mathrm{opt}}=\mathbf{R}_{\mathbf{rr}}^{-1}\mathbf{R}_{\mathbf{r}x}$.
\end{proof}

The subsequent corollary demonstrates that the mixed-ADC architecture achieves better performance than antenna selection with the same number of high-resolution ADCs.

\begin{cor}
\label{cor:cor_4}
Suppose that the high-resolution ADCs are switched to the antennas with the strongest $K$ link magnitude gains, and denote the corresponding ADC switch vector as $\bm{\delta}^{'}$. Then, the following relationship
\begin{equation}
I_{\mathrm{GMI}}(\mathbf{w}_{\mathrm{opt}},\bm{\delta}^{'})>C(N,K,0)
\label{equ:equ_60}
\end{equation}
holds, where $C(N,K,0)=\log(1+\sum_{n=1}^N\delta_{n}^{'}\cdot |h_n|^2\mathcal{E}_\mathrm{s})$ is the capacity of the antenna selection solution.
\end{cor}

\begin{proof}
Provided that the high-resolution ADCs are switched according to $\bm{\delta}^{'}$, by specifying $w_n=\delta_n^{'}\cdot h_n$, $n=1,...,N$, it is straightforward to verify that $I_{\mathrm{GMI}}(\mathbf{w},\bm{\delta}^{'})=C(N,K,0)$. Since this choice of $\mathbf{w}$ is not optimal, we have $I_{\mathrm{GMI}}(\mathbf{w}_{\mathrm{opt}},\bm{\delta}^{'})>I_{\mathrm{GMI}}(\mathbf{w},\bm{\delta}^{'})$ and \eqref{equ:equ_60} follows.
\end{proof}

When $K=N$, i.e., all the $N$ pairs of ADCs are high-resolution, we have the following corollary of Proposition \ref{prop:prop_3}.

\begin{cor}
\label{cor:cor_1}
For the special case of $K=N$, the optimal linear combiner \eqref{equ:equ_21} reduces to a maximum ratio combiner (MRC). Thus in this case, the GMI coincides with the channel capacity of conventional architecure $C(N,N,0)$.
\end{cor}

\begin{proof}
For the special case of $K=N$, i.e., $\bm{\delta}=\mathbf{1}$, \eqref{equ:equ_18} reduces to $\mathbf{R}_{\mathbf{r}x}=\mathcal{E}_\mathrm{s}\mathbf{h}$, and \eqref{equ:equ_19} reduces to
$\mathbf{R}_{\mathbf{rr}}=\mathbf{I}+\mathcal{E}_\mathrm{s}\mathbf{h}\mathbf{h}^H$. Then, the optimal combiner \eqref{equ:equ_21} turns out to be an MRC, since
\begin{equation}
\mathbf{w}_{\mathrm{opt}}=\mathbf{R}_{\mathbf{rr}}^{-1}\mathbf{R}_{\mathbf{r}x}
=\frac{\mathcal{E}_\mathrm{s}}{1+\mathcal{E}_\mathrm{s}\|\mathbf{h}\|^2}\mathbf{h}.
\label{equ:equ_26}
\end{equation}
Consequently, it is straightforward to verify that the effective SNR in \eqref{equ:equ_9} is
\begin{equation}
\frac{\kappa(\mathbf{w}_{\mathrm{opt}},\bm{\delta})}{1-\kappa(\mathbf{w}_{\mathrm{opt}},\bm{\delta})}=\|\mathbf{h}\|^2\mathcal{E}_\mathrm{s},
\label{equ:equ_27}
\end{equation}
thus completing the proof.
\end{proof}

\subsection{Asymptotic Behaviors of $I_{\mathrm{GMI}}(\mathbf{w}_{\mathrm{opt}},\bm{\delta})$}
\label{subsect:asymptotic behaviors}
In the previous subsection, the optimal linear combiner for the mixed-ADC architecture is derived. Thus we are ready to examine its asymptotic performance in both low and high SNR regimes. Letting SNR tend to zero, we have the following corollary.

\begin{cor}
\label{cor:cor_2}
As $\mathcal{E}_\mathrm{s}\rightarrow 0$, for given $\bm{\delta}$ we have
\begin{equation}
I_{\mathrm{GMI}}(\mathbf{w}_{\mathrm{opt}},\bm{\delta})
=\sum_{n=1}^{N}\left(\delta_n+\bar{\delta}_n\cdot\frac{2}{\pi}\right)|h_n|^2\mathcal{E}_\mathrm{s}+o(\mathcal{E}_\mathrm{s}).
\label{equ:equ_28}
\end{equation}
\end{cor}

See Appendix-B for its proof. Comparing with $C(N,N,0)$ in the low SNR regime, i.e., $C(N,N,0)$ $=\sum_{n=1}^{N} |h_n|^2\mathcal{E}_\mathrm{s}+o(\mathcal{E}_\mathrm{s})$, we conclude that part of the achievable rate is degraded by a factor of $\frac{2}{\pi}$ due to one-bit quantization. The expression \eqref{equ:equ_28} also suggests that, in the low SNR regime, high-resolution ADCs should be switched to the antennas with the strongest $K$ link magnitude gains.

For the high SNR case, the subsequent corollary collects our results.

\begin{cor}
\label{cor:cor_3}
As $\mathcal{E}_\mathrm{s}\rightarrow\infty$, for given $\bm{\delta}$ we have the effective SNR in \eqref{equ:equ_9} as
\begin{equation}
\frac{\kappa(\mathbf{w}_{\mathrm{opt}},\bm{\delta})}{1-\kappa(\mathbf{w}_{\mathrm{opt}},\bm{\delta})}=
\|\mathbf{p}\|^2\mathcal{E}_\mathrm{s}
+\frac{[4+O(1/\mathcal{E}_\mathrm{s})]\mathbf{q}^{H}\mathbf{B}^{-1}\mathbf{q}}
{\pi-[4+O(1/\mathcal{E}_\mathrm{s})]\mathbf{q}^{H}\mathbf{B}^{-1}\mathbf{q}},
\label{equ:equ_29}
\end{equation}
with $\mathbf{p}$, $\mathbf{q}$, and $\mathbf{B}$ given in \eqref{equ:app_21} and \eqref{equ:app_25}. As a result, $I_{\mathrm{GMI}}(\mathbf{w}_{\mathrm{opt}},\bm{\delta})$ scales as
\begin{equation}
I_{\mathrm{GMI}}(\mathbf{w}_{\mathrm{opt}},\bm{\delta})=2\log\|\mathbf{p}\|+\log(\mathcal{E}_\mathrm{s})+O(1/\mathcal{E}_\mathrm{s}).
\label{equ:equ_30}
\end{equation}
Besides, for the special case of pure one-bit quantization, i.e., $K=0$, we get
\begin{equation}
\lim_{\mathcal{E}_\mathrm{s}\rightarrow\infty}I_{\mathrm{GMI}}(\mathbf{w}_{\mathrm{opt}},\bm{\delta})
=\log\left(1+\frac{4\mathbf{q}^{H}\mathbf{B}^{-1}\mathbf{q}}{\pi-4\mathbf{q}^{H}\mathbf{B}^{-1}\mathbf{q}}\right),
\label{equ:equ_31}
\end{equation}
where $\mathbf{B}$ is also given by \eqref{equ:app_25} suppressing all the $O(1/\mathcal{E}_\mathrm{s})$ terms.
\end{cor}

The proof is given in Appendix-C. From \eqref{equ:equ_29} we notice that the contributions of high-resolution ADCs and one-bit ADCs in the high SNR regime are separate, as the first term corresponding to high-resolution ADCs increases linearly with $\mathcal{E}_\mathrm{s}$, whereas the second term coming from one-bit ADCs tends to a positive constant independent of $\mathcal{E}_\mathrm{s}$. Comparing with Corollary \ref{cor:cor_2}, we infer that one-bit ADCs are getting less beneficial as the SNR grows large, as will be validated by numerical study in Section \ref{sect:numerical}. In addition to these, \eqref{equ:equ_30} suggests for high SNR that, high-resolution ADCs should also be switched to the antennas with the strongest $K$ link magnitude gains.

For the special case of pure one-bit quantization, \eqref{equ:equ_31} indicates that the corresponding GMI approaches a finite limit, and thus the rate loss due to one-bit quantization is substantial. This is much different from the conclusion we get in the low SNR regime, where one-bit quantization degrades the achievable rate only by a factor of $\frac{2}{\pi}$. The reason underlying this phenomenon is that the amplitude of the transmit signal cannot be recovered at the receiver when $\mathcal{E}_\mathrm{s}$ is sufficiently large, and thus further enhancing the SNR does not help in improving $I_{\mathrm{GMI}}(\mathbf{w}_{\mathrm{opt}},\bm{\delta})$ (see also \cite{jacobsson2015one} \cite{knudson14arxiv} for similar phenomena).

\subsection{Performance Improvement via Dithering}
\label{subsect:dithering}
In the previous part of this section, we derived the optimal linear combiner and explored the asymptotic behaviors of $I_{\mathrm{GMI}}(\mathbf{w}_{\mathrm{opt}},\bm{\delta})$ in both low and high SNR regimes. As will be revealed by the corresponding numerical study in Section \ref{sect:numerical}, increasing $\mathrm{SNR}$ may indeed degrade the GMI when the SNR exceeds a certain threshold that depends on a collection of system parameters. In this situation, Gaussian noise, as a special type of dither, can expand the effective bit-width of one-bit ADCs and thus helps reduce the estimation bias \cite{gustavsson2014impact} \cite{dabeer2006signal}. Uniform dithering is known to be asymptotically optimal under certain problem setups \cite{dabeer2006signal}, but its non-asymptotic analysis is not amenable to analysis. Therefore, we adopt Gaussian dithering and investigate its impact on the system performance.

We consider a dithering strategy, which injects additional Gaussian noise into the antenna output before quantization when the corresponding pair of ADCs are one-bit and the receive SNR of the antenna, $|h_n|^2\mathcal{E}_\mathrm{s}$, exceeds a prescribed threshold $\mathcal{T}$. The power of the injected Gaussian noise is adjusted so that the resulting receive SNR of this antenna after dithering is pulled back to $\mathcal{T}$. Accordingly, we rewrite \eqref{equ:equ_2} as
\begin{equation}
r_n=
\begin{cases}
h_n x+z_n,\ &\mathrm{if}\ \delta_n=1,\\
\mathrm{sgn}(h_n x+z_n),\ &\mathrm{if}\ \delta_n=0,\ |h_n|^2\mathcal{E}_\mathrm{s}\leq\mathcal{T},\\
\mathrm{sgn}(h_n x+z_n+z_n^{\mathrm{d}}),\ &\mathrm{if}\ \delta_n=0,\ |h_n|^2\mathcal{E}_\mathrm{s}>\mathcal{T},
\end{cases}
\label{equ:equ_32}
\end{equation}
where the Gaussian dither $z_n^{\mathrm{d}}\sim\mathcal{CN}(0,|h_n|^2\mathcal{E}_\mathrm{s}/\mathcal{T}-1)$ is independent of $z_n$ so that $z_n+z_n^{\mathrm{d}}\sim\mathcal{CN}(0,|h_n|^2\mathcal{E}_\mathrm{s}/\mathcal{T})$. Since high SNR is always favorable for high-resolution ADC, we do not perform dithering for antennas with high-resolution ADCs.

The system architecture and optimal linear combiner developed in Section \ref{sect:GMI and combiner} still apply directly, except that we need to make some modifications about $\mathbf{R}_{\mathbf{r}x}$ in \eqref{equ:equ_18} and $\mathbf{R}_{\mathbf{rr}}$ in \eqref{equ:equ_19}: for any $n\in\{1,2,...,N\}$, whenever $\delta_n=0$ and $|h_n|^2\mathcal{E}_\mathrm{s}>\mathcal{T}$, we make the following substitution,
\begin{equation}
|h_n|^2\mathcal{E}_\mathrm{s}+1\longrightarrow|h_n|^2\mathcal{E}_\mathrm{s}(1+1/\mathcal{T}),
\label{equ:equ_33}
\end{equation}
in \eqref{equ:equ_18} and \eqref{equ:equ_19}. The optimal threshold $\mathcal{T}_{\mathrm{opt}}$ depends on $K$, $N$, and $\mathrm{SNR}$. For the situation with relatively small $K$, the dependence of $\mathcal{T}_{\mathrm{opt}}$ on $K$ is actually negligible. Nevertheless, the analytical optimization of $\mathcal{T}$ is still difficult, and thus we perform a numerical search. To be specific, for any given $\mathrm{SNR}$ and $N$, we find the optimal threshold $\mathcal{T}_{\mathrm{opt}}$ for $K=0$ through a Monte Carlo simulation, and then use $\mathcal{T}_{\mathrm{opt}}$ to evaluate the performance gain with $K \geq 1$ as well. Numerical results will be presented in Section \ref{sect:numerical}.

\section{Ergodic Fading Channels}
\label{sect:ergodic fading channel}
Although our analysis thus far has been for the fixed channel scenario, the analytical framework developed can be extended to the the randomly varying channel scenario. We assume that the channel fading process $\{\mathbf{h}^{l}\}$ obeys the block fading channel model among coherence intervals. We start with the perfect CSI situation and then investigate the impact of channel estimation error on performance.

\subsection{Perfect CSI}
Since the channel vector $\mathbf{h}$ varies over time now, $\mathbf{w}$ and $\bm{\delta}$ in this situation shall be designed based on the instantaneous channel realization. In this situation, the GMI becomes\footnote{Here for simplicity we consider a fixed value of $a$ in the nearest neighbor decoding metric. Allowing $a$ to vary based on $\mathbf{h}^l$ may result in some performance improvement especially when $N$ is not too large.}
\begin{align}
I_{\mathrm{GMI}}=&\sup_{a\in\mathbb{C},\theta<0}
\Bigg(\theta\mathbb{E}_{x,\mathbf{z},\mathbf{h}}[|f(x,\mathbf{h},\mathbf{z})-ax|^2]-\nonumber\\
&\frac{\theta\mathbb{E}_{x,\mathbf{z},\mathbf{h}}[|f(x,\mathbf{h},\mathbf{z})|^2]}{1-\theta|a|^2\mathcal{E}_\mathrm{s}}
\!+\!\log(1-\theta|a|^2\mathcal{E}_\mathrm{s})\Bigg).
\label{equ:equ_50}
\end{align}
Notice that it shares the same nominal form as \eqref{equ:equ_49} except that the expectation here is over $x$, $\mathbf{z}$, and $\mathbf{h}$. Recognizing the difficulty of this optimization problem, we turn to evaluate the lower and upper bounds of $I_{\mathrm{GMI}}$, and arrive at the following proposition. Numerical results will be given in Section \ref{sect:numerical} to verify the tightness of the lower and upper bounds.

\begin{prop}
For the ergodic fading channel scenario, lower and upper bounds of $I_{\mathrm{GMI}}$ are given by
\begin{eqnarray}
I_{\mathrm{GMI}}^{\mathrm{lower}}&=&\log
\left(1+\frac{\mathbb{E}_{\mathbf{h}}[\kappa(\mathbf{w}_{\mathrm{opt}},\bm{\delta})]}{1-\mathbb{E}_{\mathbf{h}}[\kappa(\mathbf{w}_{\mathrm{opt}},\bm{\delta})]}\right),\\
I_{\mathrm{GMI}}^{\mathrm{upper}}&=&\mathbb{E}_{\mathbf{h}}
\left[\log\left(1+\frac{\kappa(\mathbf{w}_{\mathrm{opt}},\bm{\delta})}{1-\kappa(\mathbf{w}_{\mathrm{opt}},\bm{\delta})}\right)\right],\ \
\label{equ:equ_52}
\end{eqnarray}
respectively, where $\kappa(\mathbf{w}_{\mathrm{opt}},\bm{\delta})$ is given by \eqref{equ:equ_22}.
\label{prop:prop_6}
\end{prop}

\begin{proof}
Following a similar procedure as \cite[App. C]{zhang2012general}, we obtain $\kappa$ in this situation as
\begin{equation}
\kappa=\frac{|\mathbb{E}_{x,\mathbf{z},\mathbf{h}}[f^*(x,\mathbf{h},\mathbf{z})\cdot x]|^2}
{\mathcal{E}_\mathrm{s}\mathbb{E}_{x,\mathbf{z},\mathbf{h}}[|f(x,\mathbf{h},\mathbf{z})|^2]},
\label{equ:equ_53}
\end{equation}
which shares exactly the same form as \eqref{equ:equ_10}, except that the expectation is taken over $x$, $\mathbf{z}$, and $\mathbf{h}$. The maximization of $\kappa$ shall be accomplished by optimizing the linear combiner. Therefore by specifying $\mathbf{w}$ to be designed according to \eqref{equ:equ_21}, we get a lower bound of the optimal $\kappa$, since this design is just one of the feasible options and thus is not necessarily optimal; that is
\begin{eqnarray}
\kappa&=&\frac{|\mathbb{E}_{\mathbf{h}}[\mathbb{E}_{x,\mathbf{z}}[f^*(x,\mathbf{h},\mathbf{z})\cdot x|\mathbf{h}]]|^2}
{\mathcal{E}_\mathrm{s}\mathbb{E}_{\mathbf{h}}[\mathbb{E}_{x,\mathbf{z}}[|f(x,\mathbf{h},\mathbf{z})|^2|\mathbf{h}]]}\nonumber\\
&\geq&\frac{|\mathbb{E}_{\mathbf{h}}[\mathbf{w}^H\mathbf{R}_{\mathbf{r}x}]|^2}
{\mathcal{E}_\mathrm{s}\mathbb{E}_{\mathbf{h}}[\mathbf{w}^{H}\mathbf{R}_{\mathbf{rr}}\mathbf{w}]}\nonumber\\
&=&\mathbb{E}_{\mathbf{h}}[\kappa(\mathbf{w}_{\mathrm{opt}},\bm{\delta})],
\label{equ:equ_54}
\end{eqnarray}
where the last equation comes from \eqref{equ:equ_21}-\eqref{equ:equ_22}. Consequently, we obtain the lower bound of $I_{\mathrm{GMI}}$ as given by (29).

To prove \eqref{equ:equ_52}, we first rewrite \eqref{equ:equ_50} as
\begin{align}
I_{\mathrm{GMI}}=&\sup_{a\in\mathbb{C},\theta<0}\mathbb{E}_{\mathbf{h}}
\Bigg(\theta\mathbb{E}_{x,\mathbf{z}}[|f(x,\mathbf{h},\mathbf{z})-ax|^2|\mathbf{h}]-\nonumber\\
&\frac{\theta\mathbb{E}_{x,\mathbf{z}}[|f(x,\mathbf{h},\mathbf{z})|^2|\mathbf{h}]}
{1-\theta|a|^2\mathcal{E}_\mathrm{s}}
+\log(1-\theta|a|^2\mathcal{E}_\mathrm{s})\Bigg).
\label{equ:equ_59}
\end{align}
Then, to derive the upper bound we simply exchange the order of supremum operation and the expectation over $\mathbf{h}$. This leads to
\begin{align}
I_{\mathrm{GMI}}\leq\ &\mathbb{E}_{\mathbf{h}}\Bigg(\sup_{a\in\mathbb{C},\theta<0}
\Big(\theta\mathbb{E}_{x,\mathbf{z}}[|f(x,\mathbf{h},\mathbf{z})-ax|^2|\mathbf{h}]-\nonumber\\
&\frac{\theta\mathbb{E}_{x,\mathbf{z}}[|f(x,\mathbf{h},\mathbf{z})|^2|\mathbf{h}]}
{1-\theta|a|^2\mathcal{E}_\mathrm{s}}
\!+\!\log(1\!-\!\theta|a|^2\mathcal{E}_\mathrm{s})\Big)\Bigg).
\label{equ:equ_55}
\end{align}
Consequently, \eqref{equ:equ_52} follows directly from \eqref{equ:equ_49} and the subsequent results established for the fixed SIMO channels.
\end{proof}

\subsection{Training and Effect of Imperfect CSI}
Our results derived thus far are based on the perfect CSI assumption. In practice, however, CSI needs to be either explicitly or implicitly acquired, say, via channel estimation. The channel estimation procedure with coarsely quantized channel outputs is both inefficient and elusive for analysis. Therefore, to study the robustness of the mixed-ADC architecture to imperfect CSI, in this paper we only utilize the high-resolution ADCs to perform channel estimation.

Specifically, we estimate the channel vector in a round-robin manner, by which we link the $K$ pairs of high-resolution ADCs to the first $K$ antennas and estimate the corresponding channel coefficients $h_1,...,h_K$ at the first symbol time, turn the $K$ pairs of high-resolution ADCs to the next $K$ antennas and estimate $h_{K+1},...,h_{2K}$ at the next symbol time, and so on. Thus the training phase lasts about $N/K$ symbol times\footnote{For example, a BS equipped with 100 antennas and 20 pairs of high-resolution ADCs would consume 5 symbol times in each coherence interval for channel estimation. This overhead is acceptable for slowly or moderately varying fading channels; for example, in \cite{ngo2013energy} the channel coherence interval length is taken as 196, which is also used by us in the subsequent simulations. The efficiency and quality of channel training may be improved by jointly exploiting high-resolution ADCs and one-bit ADCs, which is an interesting and important topic for future research.}. To simplify analysis, in this subsection we assume that each antenna follows i.i.d. Rayleigh fading, so that ${h}_n\sim\mathcal{CN}(0,1)$, $n = 1, 2, \ldots, N$. An MMSE estimator is adopted at the BS, and thus without loss of generality, we can decompose $h_n$ into
\begin{equation}
h_n=\hat{h}_n+\tilde{h}_n,\ \ n=1,...,N,
\label{equ:equ_34}
\end{equation}
where $\hat{h}_n\sim\mathcal{CN}(0,1-\sigma_\mathrm{t}^2)$ is the estimated channel coefficient, while $\tilde{h}_n\sim\mathcal{CN}(0,\sigma_\mathrm{t}^2)$ accounts for the independent estimation error. Accordingly, we define the MSE of the channel estimation as $\mathrm{MSE}_{\mathrm{t}}=\sigma_\mathrm{t}^2$.

In this situation, the linear combiner $\mathbf{w}$ and the ADC switch vector $\bm{\delta}$ should be designed based on the channel estimate $\hat{\mathbf{h}}$. Besides, we rewrite $f(x,\mathbf{h},\mathbf{z})$ as $f(x,\hat{\mathbf{h}},\tilde{\mathbf{h}},\mathbf{z})$ in order to incorporate the effect of channel estimation. Then with some modification, our analysis developed in the last subsection still applies for the imperfect CSI case. To proceed, we have
\begin{align}
I_{\mathrm{GMI}}^{\mathrm{im}}=&\frac{T - N/K}{T} \sup_{a\in\mathbb{C},\theta<0}
\Bigg(\theta\mathbb{E}[|f(x,\hat{\mathbf{h}},\tilde{\mathbf{h}},\mathbf{z})-ax|^2]\nonumber\\
&-\frac{\theta\mathbb{E}[|f(x,\hat{\mathbf{h}},\tilde{\mathbf{h}},\mathbf{z})|^2]}{1-\theta|a|^2\mathcal{E}_\mathrm{s}}
+\log(1-\theta|a|^2\mathcal{E}_\mathrm{s})\Bigg),
\label{equ:equ_56}
\end{align}
which obeys an analogous form as \eqref{equ:equ_50}, except that the leading coefficient $\frac{T - N/K}{T}$ accounts for the rate loss due to channel training ($T$ is the coherence interval length), and that the expectation here is taken with respect to $x$, $\hat{\mathbf{h}}$, $\tilde{\mathbf{h}}$, and $\mathbf{z}$. Exploiting a similar argument as that in the proof of Proposition \ref{prop:prop_6}, we arrive at the following proposition.

\begin{prop}
For block fading channels with imperfect CSI, a lower bound of $I_{\mathrm{GMI}}^{\mathrm{im}}$ is
\begin{equation}
I_{\mathrm{GMI}}^{\mathrm{im,l}}=\frac{T - N/K}{T} \log\left(1+\frac{\mathbb{E}_{\hat{\mathbf{h}}}[\kappa(\mathbf{w}_{\mathrm{opt}}^{\mathrm{im}},\bm{\delta})]}
{1-\mathbb{E}_{\hat{\mathbf{h}}}[\kappa(\mathbf{w}_{\mathrm{opt}}^{\mathrm{im}},\bm{\delta})]}\right),
\label{equ:equ_57}
\end{equation}
and an upper bound of $I_{\mathrm{GMI}}^{\mathrm{im}}$ is
\begin{equation}
I_{\mathrm{GMI}}^{\mathrm{im,u}}=\frac{T - N/K}{T}\mathbb{E}_{\hat{\mathbf{h}}}
\left[\log\left(1+
\frac{\kappa(\mathbf{w}_{\mathrm{opt}}^{\mathrm{im}},\bm{\delta})}
{1-\kappa(\mathbf{w}_{\mathrm{opt}}^{\mathrm{im}},\bm{\delta})}\right)\right].
\label{equ:equ_58}
\end{equation}
Here, $\mathbf{w}_{\mathrm{opt}}^{\mathrm{im}}$ and $\kappa(\mathbf{w}_{\mathrm{opt}}^{\mathrm{im}},\bm{\delta})$ also come from \eqref{equ:equ_21} and \eqref{equ:equ_22}, but we need to replace $\mathbf{R}_{\mathbf{r}x}$ with $\mathbf{R}_{\mathbf{r}x}^{\mathrm{im}}=\mathbb{E}_{\tilde{\mathbf{h}}}[\mathbf{R}_{\mathbf{r}x}]$, and replace $\mathbf{R}_{\mathbf{rr}}$ with $\mathbf{R}_{\mathbf{rr}}^{\mathrm{im}}=\mathbb{E}_{\tilde{\mathbf{h}}}[\mathbf{R}_{\mathbf{rr}}]$.
\label{prop:prop_4}
\end{prop}

\section{Extension to multi-user scenario}
\label{sect:multi-user}
In this section, we consider a multi-user system where the BS serves $M$ single-antenna users simultaneously. The CSI is assumed perfectly known by the BS, and there are still only $K$ pairs of high-resolution ADCs available.

\subsection{Fixed Channels}
Again, we start from the fixed channel case. The channel matrix between the users and the BS is denoted by $\mathbf{H}\triangleq[\mathbf{h}_1,...,\mathbf{h}_N]\in\mathbb{C}^{M\times N}$, i.e., $\mathbf{h}_n\triangleq[h_{1n},...,h_{Mn}]^T$ collecting the channel coefficients related to the $n$-th antenna at the BS. We write the quantized output at the $n$-th antenna, with user $j$ considered, as
\begin{equation}
r_n^{\mathrm{mu}}=\delta_n\!\cdot\!\left(\sum_{\iota=1}^{M}h_{\iota n}x_\iota+z_n\right)+\bar{\delta}_n\!\cdot\!\mathrm{sgn}\left(\sum_{\iota=1}^{M}h_{\iota n}x_\iota+z_n\right),
\label{equ:equ_43}
\end{equation}
where $x_{\iota}\sim\mathcal{CN}(0,\mathcal{E}_\mathrm{s})$ denotes the i.i.d. coded signal dedicated to the $\iota$-th user, and $\sum_{\iota\neq j}^{M}h_{\iota n}x_\iota+z_n$ summarizes the co-channel interference and noise for the considered user $j$. For a fair comparison, the SNR in this situation is defined as $\mathrm{SNR}=M\mathcal{E}_\mathrm{s}$, reflecting the total transmit power from all the users.

Following a similar derivation procedure as that in Section \ref{sect:GMI and combiner}, we get the GMI of the considered user. The proof is omitted for concision.

\begin{prop}
\label{prop:prop_5}
For given $\mathbf{H}$ and $\bm{\delta}$, when treating other users' signals as noise, the GMI of user $j$ is
\begin{equation}
I_{\mathrm{GMI}}^{\mathrm{mu}}=\log\left(1+\frac{\kappa^{\mathrm{mu}}}{1-\kappa^{\mathrm{mu}}}\right),
\label{equ:equ_44}
\end{equation}
where the parameter $\kappa^{\mathrm{mu}}$ is
\begin{equation}
\kappa^{\mathrm{mu}}=\frac{1}{\mathcal{E}_\mathrm{s}}
(\mathbf{R}_{\mathbf{r}x}^{\mathrm{mu}})^{H}(\mathbf{R}_{\mathbf{rr}}^{\mathrm{mu}})^{-1}\mathbf{R}_{\mathbf{r}x}^{\mathrm{mu}}.
\label{equ:equ_45}
\end{equation}
$\mathbf{R}_{\mathbf{r}x}^{\mathrm{mu}}$ is the correlation vector between $\mathbf{r}^{\mathrm{mu}}$ and $x_j$, with its $n$-th entry given as
\begin{equation}
(\mathbf{R}_{\mathbf{r}x}^{\mathrm{mu}})_n=h_{jn} \mathcal{E}_\mathrm{s} \left[\delta_n+\bar{\delta}_n\cdot\sqrt{\frac{4}{\pi(\|\mathbf{h}_n\|^2\mathcal{E}_\mathrm{s}+1)}}\right],
\label{equ:equ_46}
\end{equation}
and $\mathbf{R}_{\mathbf{rr}}^{\mathrm{mu}}$ is the covariance matrix of $\mathbf{r}^{\mathrm{mu}}$, with the $(n,m)$-th entry being $(\mathbf{R}_{\mathbf{rr}}^{\mathrm{mu}})_{n,m}=$
\begin{equation}
\begin{cases}
1+\delta_n\cdot\|\mathbf{h}_n\|^2\mathcal{E}_\mathrm{s}+\bar{\delta}_n,\!\!\!\!\!&\mathrm{if}\ n=m, \\
\mathbf{h}_n^T\mathbf{h}_m^* \mathcal{E}_\mathrm{s}
\Bigg[\delta_n\delta_m+\delta_n\bar{\delta}_m\cdot\sqrt{\frac{4}{\pi(\|\mathbf{h}_m\|^2\mathcal{E}_\mathrm{s}+1)}}+\\
\ \ \ \ \ \ \ \ \ \ \ \bar{\delta}_n\delta_m\cdot\sqrt{\frac{4}{\pi(\|\mathbf{h}_n\|^2\mathcal{E}_\mathrm{s}+1)}}\Bigg]+\\
\bar{\delta}_n\bar{\delta}_m\!\cdot\!\frac{4}{\pi}\Bigg[
\mathrm{arcsin}\Big(\frac{(\mathbf{h}_n^T\mathbf{h}_m^*)^{\mathrm{R}}\mathcal{E}_\mathrm{s}}
{\sqrt{\|\mathbf{h}_n\|^2\mathcal{E}_\mathrm{s}+1}\sqrt{\|\mathbf{h}_m\|^2\mathcal{E}_\mathrm{s}+1}}\Big)+\\
\ \ \ \ \ \ \ \ \ \ i\!\cdot\!\mathrm{arcsin}\Big(\frac{(\mathbf{h}_n^T\mathbf{h}_m^*)^{\mathrm{I}}\mathcal{E}_\mathrm{s}}
{\sqrt{\|\mathbf{h}_n\|^2\mathcal{E}_\mathrm{s}+1}\sqrt{\|\mathbf{h}_m\|^2\mathcal{E}_\mathrm{s}+1}}\Big)
\Bigg],\!\!\!\!\!&\mathrm{if}\ n\neq m.
\end{cases}
\label{equ:equ_47}
\end{equation}
\end{prop}

In the multi-user scenario, there is no clear clue about how to switch the high-resolution ADCs. To obtain some hint, we explore the asymptotic behavior of \eqref{equ:equ_44} in the low SNR regime, leading to the corollary below.

\begin{cor}
\label{cor:cor_5}
When $\mathcal{E}_\mathrm{s}\rightarrow 0$, for given $\mathbf{H}$ and $\bm{\delta}$, we have the GMI of user $j$ as
\begin{equation}
I_{\mathrm{GMI}}^{\mathrm{mu}}
=\sum_{n=1}^{N}\left(\delta_n+\bar{\delta}_n\cdot\frac{2}{\pi}\right)\cdot|h_{jn}|^2\mathcal{E}_\mathrm{s}+o(\mathcal{E}_\mathrm{s}).
\label{equ:equ_48}
\end{equation}
\end{cor}

The proof procedure is virtually the same as Appendix-B and thus is omitted. We notice that $I_{\mathrm{GMI}}^{\mathrm{mu}}$ behaves analogously with $I_{\mathrm{GMI}}$ in the low SNR regime, which is foreseeable as the system is now noise-limited. The sum GMI now equals $\sum_{n=1}^N\sum_{j=1}^M\left(\delta_n+\bar{\delta}_n\cdot\frac{2}{\pi}\right)\cdot|h_{jn}|^2\mathcal{E}_\mathrm{s}+o(\mathcal{E}_\mathrm{s})$, which suggests that the $K$ pairs of high-resolution ADCs may be switched to the antennas with the $K$ largest $\sum_{j=1}^{M}|h_{jn}|^2$.

The asymptotic behavior of $I_{\mathrm{GMI}}^{\mathrm{mu}}$ in the high SNR regime is analytically intractable, and thus there is no generally convincing ADC switch scheme for the multi-user scenario. For this reason, we consider two heuristic switch schemes in the numerical study.
\begin{itemize}
\item Random switch: high-resolution ADCs are switched randomly.
\item Norm-based switch: as suggested by Corollary \ref{cor:cor_5}, high-resolution ADCs are switched to antennas with the $K$ largest $\sum_{j=1}^{M}|h_{jn}|^2$.
\end{itemize}
Numerical results will be given in Section \ref{sect:numerical} to examine the performance of both switch schemes.

\subsection{Ergodic Fading Channels}
The analysis is then naturally applied to ergodic fading channels, as summarized by the subsequent proposition. Numerical study will also be conducted in Section \ref{sect:numerical} to verify the tightness of the lower and upper bounds.
\begin{prop}
\label{prop:prop_7}
For ergodic fading channels, lower and upper bounds of the GMI for user $j$ are
\begin{eqnarray}
I_{\mathrm{GMI}}^{\mathrm{mu,l}}&=&\log\left(1+\frac{\mathbb{E}_{\mathbf{H}}[\kappa^{\mathrm{mu}}]}{1-\mathbb{E}_{\mathbf{H}}[\kappa^{\mathrm{mu}}]}\right),\\
I_{\mathrm{GMI}}^{\mathrm{mu,u}}&=&\mathbb{E}_{\mathbf{H}}\left[\log\left(1+\frac{\kappa^{\mathrm{mu}}}{1-\kappa^{\mathrm{mu}}}\right)\right],
\end{eqnarray}
where the parameter $\kappa^{\mathrm{mu}}$ is given by \eqref{equ:equ_45}.
\end{prop}

\section{Energy Efficiency}
\label{sect:energy efficiency}
We establish the power models for conventional architecture (CA), antenna selection (AS), and mixed-ADC architecture (MA). Only the circuit power consumption is taken into account, since first, we focus on the receiver design, and second, the power expenditure on digital signal processing is approximately independent of the choice of receivers all of which are based on linear combining. Then power models of the three considered receivers are
\begin{eqnarray}
P_{\mathrm{CA}}&\!\!=\!\!&N(P_{\mathrm{LNA}}+P_{\mathrm{mix}}+P_{\mathrm{ADC}}+P_{\mathrm{fil}})+P_{\mathrm{syn}},\nonumber\\
P_{\mathrm{AS}}&\!\!=\!\!&K(P_{\mathrm{LNA}}+P_{\mathrm{mix}}+P_{\mathrm{ADC}}+P_{\mathrm{fil}})+P_{\mathrm{syn}},\nonumber\\
P_{\mathrm{MA}}&\!\!=\!\!&N(P_{\mathrm{LNA}}+P_{\mathrm{mix}}+P_{\mathrm{fil}})+K P_{\mathrm{ADC}}+P_{\mathrm{syn}},
\end{eqnarray}
where $P_{\mathrm{LNA}}$, $P_{\mathrm{mix}}$, $P_{\mathrm{ADC}}$, $P_{\mathrm{fil}}$, and $P_{\mathrm{syn}}$ account for the power consumption of low noise amplifier (LNA), mixer, a pair of high-resolution ADCs, filters, and frequency synthesizer (which is typically shared among all the antennas in practice), respectively. Power consumption due to one-bit ADCs is neglected, since they can be implemented as polarity detectors using discrete components and thus the power consumption is marginal compared with other parts of the circuitry.

We refer to a widely used model \cite{li2007system} to determine the power consumption parameters. Bandwidth in \cite{li2007system} is taken to be 1 MHz at a carrier frequency of $f_c=2$ GHz, while in this paper we assume a bandwidth of $B=40$ MHz\footnote{Note that LTE-Advanced supports 15-100 MHz bands in TDD uplink \cite{TR36.912}. Besides, a bandwidth of 40 MHz would be necessary for supporting an average per-user rate of 100 Mbps for future 5G.} at the same carrier frequency. To account for this scaling, realizing that the power consumption of RF front-end except ADC is insensitive to the bandwidth\footnote{See \cite{borremans20130} for example, where the signal bandwidth ranges from 0.5 MHz to 50 MHz, but the RF front-end except ADC power consumption only changes from 20 mW to 40 mW, and the change is mainly due to the fluctuation of receiver gain and noise figure.} but the power consumption of an ADC scales linearly with the bandwidth, we update the power consumption parameters as: $P_{\mathrm{LNA}}=20$ mW, $P_{\mathrm{mix}}=21$ mW, $P_{\mathrm{syn}}=67.5$ mW, $P_{\mathrm{fil}}=5$ mW, and $P_{\mathrm{ADC}}=234$ mW. As a side note, for many high-speed applications, high-resolution ADCs generally accounts for a dominant portion of the circuit power consumption; --- in some recent works (e.g., \cite{bai2015energy}), only the ADC power consumption is taken into account, ignoring the other RF front-end parts.

\begin{figure}
\centering
\includegraphics[width=0.45 \textwidth]{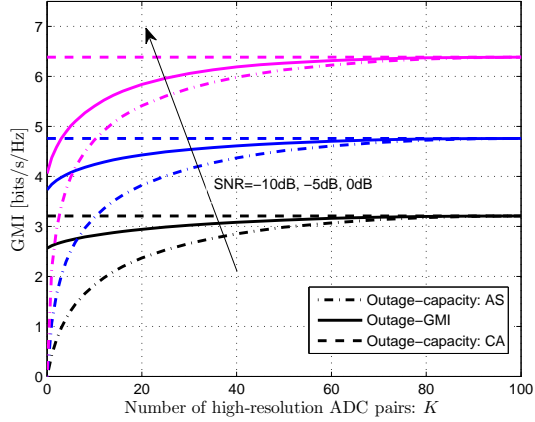}
\caption{Outage-GMI of the mixed-ADC architecture for different numbers of high-resolution ADC pairs, $N=100$, $P_{\mathrm{out}}=5\%$.}
\label{fig:fig_2}
\end{figure}
\begin{figure}
\centering
\includegraphics[width=0.45\textwidth]{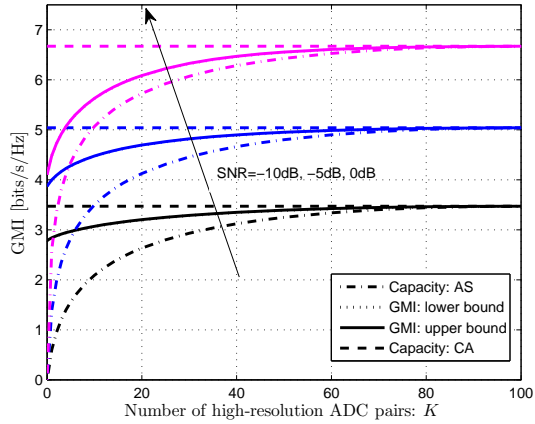}
\caption{Lower and upper bounds of the GMI for ergodic fading channels with perfect CSI at the BS, $N=100$.}
\label{fig:fig_3}
\end{figure}
\begin{figure}
\centering
\includegraphics[width=0.45\textwidth]{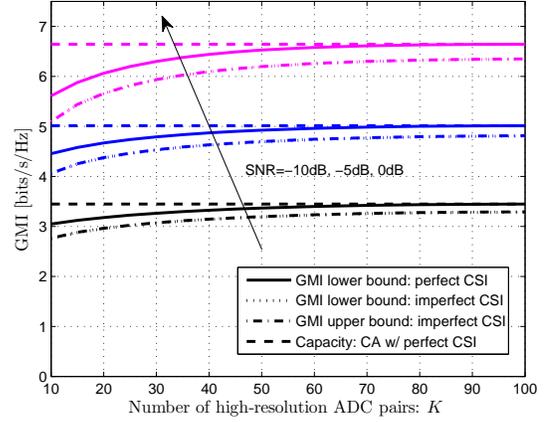}
\caption{GMI of the mixed-ADC architecture with imperfect CSI: impact of $K$, $N=100$, $T = 196$, $\mathrm{MSE}_{\mathrm{t}}=-10$dB.}
\label{fig:fig_4}
\end{figure}
\begin{figure}
\centering
\includegraphics[width=0.45\textwidth]{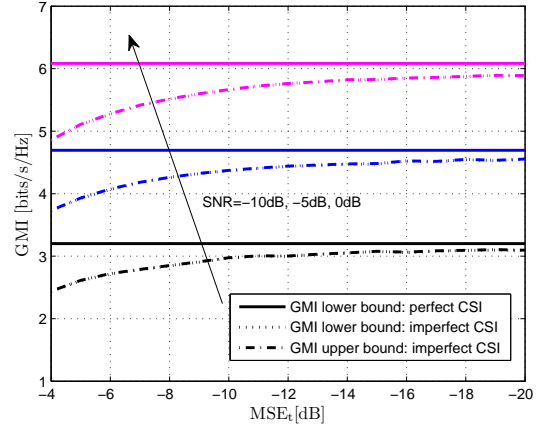}
\caption{GMI of the mixed-ADC architecture with imperfect CSI: impact of $\mathrm{MSE}_{\mathrm{t}}$, $N=100$, $T = 196$, $K=20$.}
\label{fig:fig_5}
\end{figure}
Energy efficiency is sometimes defined as the number of information bits conveyed per joule energy consumption. But this ratio alone does not capture the whole story, since the improvement of energy efficiency is valuable only if a desired spectral efficiency is ensured. For this reason, in this paper we characterize the energy efficiency using two performance metrics: normalized spectral efficiency and normalized energy consumption. Taking the mixed-ADC architecture as an example, these two performance metrics are defined as
\begin{equation}
\bar{R}_{\mathrm{MA}}=\frac{I_{\mathrm{GMI}}}{\mathbb{E}_{\mathbf{h}}[\log(1+\|\mathbf{h}\|^2\mathcal{E}_{\mathrm{s}})]},\ \ \bar{E}_{\mathrm{MA}}=\frac{P_{\mathrm{MA}}}{P_{\mathrm{CA}}},
\end{equation}
in the single-user scenario under ergodic fading. That is, we simultaneously compare the spectral efficiency and the energy efficiency of the mixed-ADC architecture against those of the conventional architecture. These performance metrics can also be straightforwardly defined for antenna selection and for multi-user systems (there the sum achievable rates are used in $\bar{R}_{\mathrm{MA}}$).

\section{Numerical Results}
\label{sect:numerical}
In this section we validate our previous analysis with numerical results. Except for the first subsection, all the results in this section are for ergodic fading channels. The channel coefficients are drawn i.i.d. from $\mathcal{CN}(0,1)$. We deem $\mathrm{SNR}$ that achieves 5 bits/s/Hz for single-user scenario or 2.5 bits/s/Hz per user for multi-user scenario as a moderate SNR \cite{TR36.912}.

\subsection{Outage-GMI for Random but Fixed SIMO Channel}
We first examine the outage performance of the mixed-ADC architecture. In this situation, the channel vector is random but fixed ever since it is chosen. Figure \ref{fig:fig_2} displays the outage-GMI\footnote{The outage-GMI is defined as the largest GMI at a specified outage probability $P_{\mathrm{out}}$. In this subsection, both the outage-GMI and the outage-capacity are obtained by running 1000 Monte Carlo simulations.} for $P_{\mathrm{out}}=5\%$. Several observations are in order. First, Figure \ref{fig:fig_2} shows that the mixed-ADC architecture with a small number of high-resolution ADCs achieves a large fraction of the outage-capacity of the conventional architecture. For example, when $\mathrm{SNR}=0\mathrm{dB}$, the mixed-ADC architecture with $K = 10$ attains 85\% of the outage-capacity of the conventional architecture, and this number rises to 92\% when $K = 20$. Besides, it indicates that one-bit ADCs are less beneficial when the SNR grows large, but significantly improve the performance in the low to moderate SNR regime, compared with antenna selection.
\begin{figure}
\centering
\includegraphics[width=0.45\textwidth]{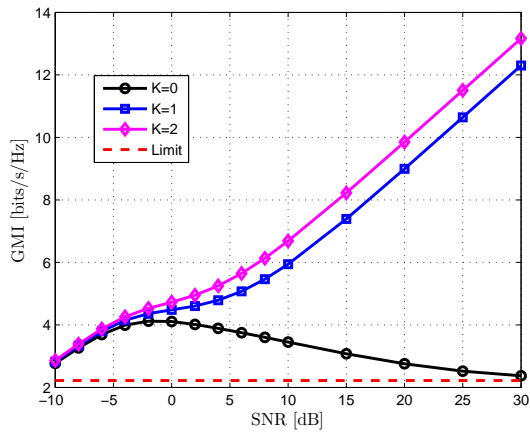}
\caption{GMI lower bound of the mixed-ADC architecture for ergodic fading channels, $N=100$.}
\label{fig:fig_6}
\end{figure}
\begin{figure}
\centering
\includegraphics[width=0.45\textwidth]{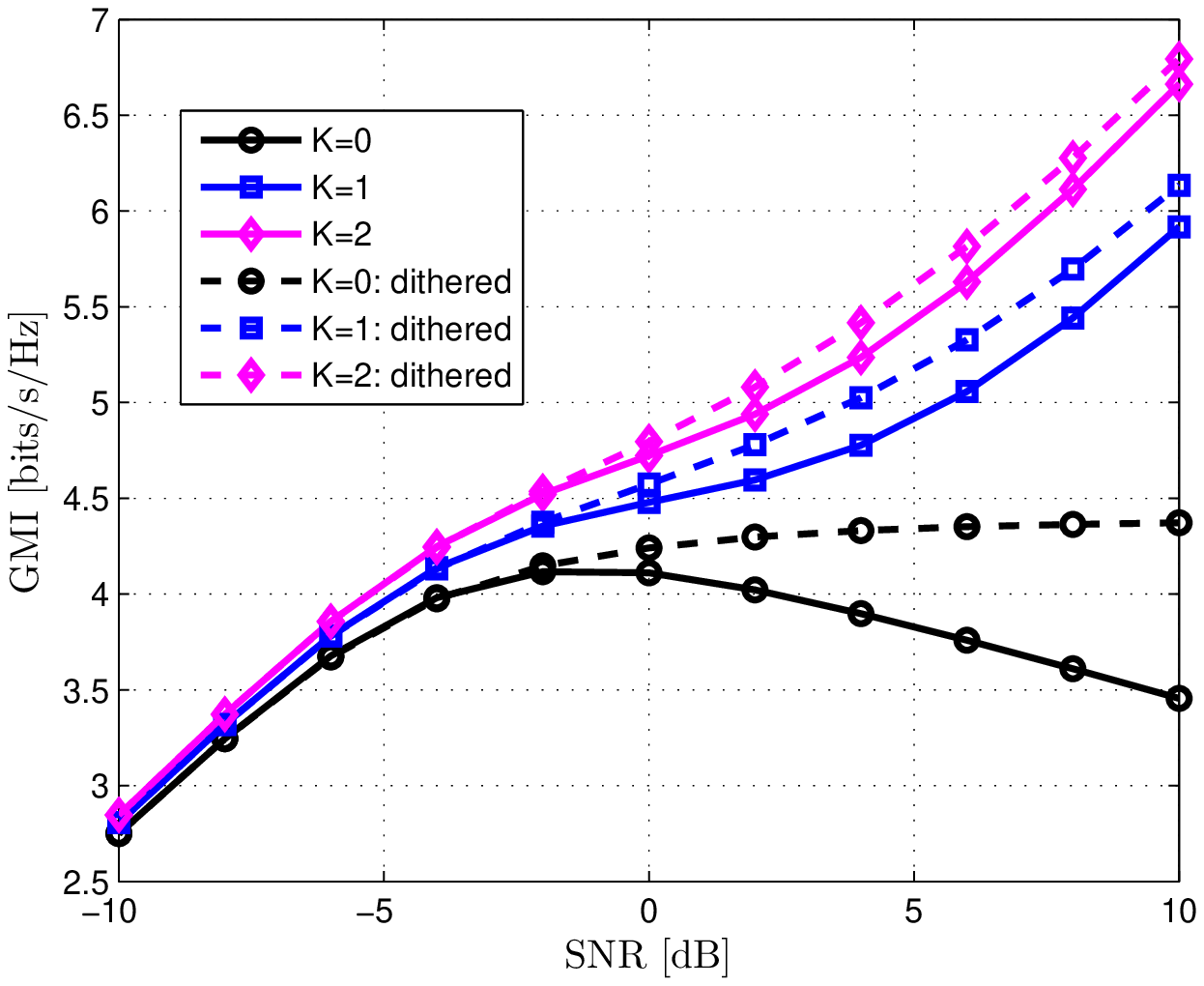}
\caption{GMI lower bound of the mixed-ADC architecture for ergodic fading channels with Gaussian dithering, $N=100$.}
\label{fig:fig_7}
\end{figure}
\begin{figure}
\centering
\includegraphics[width=0.45\textwidth]{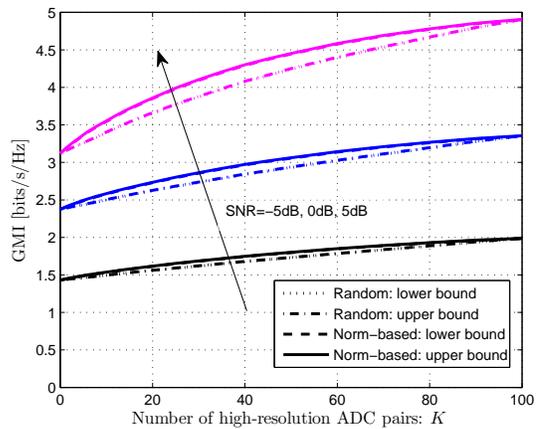}
\caption{Per-user GMI under ergodic fading: comparison of random and norm-based ADC switch schemes, $N=100$, $M=10$.}
\label{fig:fig_8}
\end{figure}
\begin{figure}
\centering
\includegraphics[width=0.45\textwidth]{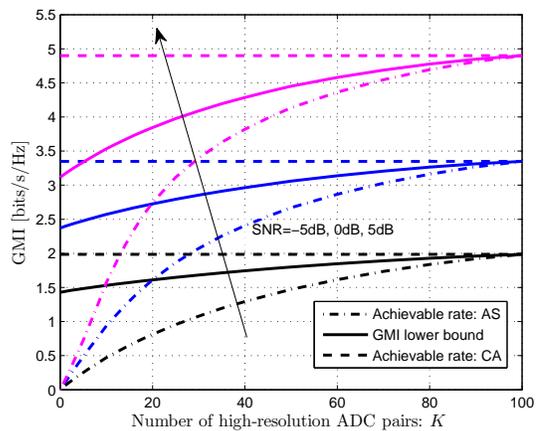}
\caption{Per-user GMI under ergodic fading: comparison with conventional architecture and antenna selection, $N=100$, $M=10$.}
\label{fig:fig_9}
\end{figure}

\subsection{GMI for Ergodic Fading SIMO Channel}
By Figure \ref{fig:fig_3}, we first examine the tightness of the lower and upper bounds derived in Proposition \ref{prop:prop_6}. It is clear that the lower and upper bounds virtually coincide with each other, and as a result, it is sufficient to use only the GMI lower bound in the following numerical study for spectral efficiency evaluation.

Then, we turn to check the impact of imperfect CSI on the performance. Numerical results are given by Figure \ref{fig:fig_4} assuming $\mathrm{MSE}_{\mathrm{t}}=-10\mathrm{dB}$, indicating that the gap between lower and upper bounds is still virtually negligible. On the other hand, though there is a noticeable rate loss due to channel estimation error, the mixed-ADC architecture with a small number of high-resolution ADCs still achieves much of the the channel capacity of the conventional architecture with perfect CSI. Besides, Figure \ref{fig:fig_5} accounts for the impact of $\mathrm{MSE}_{\mathrm{t}}$ on the performance, from which we again conclude that the mixed-ADC architecture is robust against imperfect CSI.

\subsection{Performance Gain of Gaussian Dithering}
Figure \ref{fig:fig_6} accounts for the effect of SNR on the GMI lower bound of ergodic fading channels, with special focus on small $K$. For the special case of $K=0$, we observe that $I_{\mathrm{GMI}}^{\mathrm{lower}}$ increases first but then turns downward as the SNR grows large. Besides, as predicted by Corollary \ref{cor:cor_3}, $I_{\mathrm{GMI}}^{\mathrm{lower}}$ asymptotically approaches a positive limit illustrated by the dashed line. The reason underlying this phenomenon is that the amplitude of the transmit signal cannot be recovered at the receiver when the SNR is sufficiently large with only one-bit ADCs \cite{jacobsson2015one}. With merely one pair of high-resolution ADCs, $I_{\mathrm{GMI}}^{\mathrm{lower}}$ is always increasing with $\mathrm{SNR}$, and increases linearly with respect to $10\log_{10}(\mathrm{SNR})$ in the high SNR regime as predicted by Corollary \ref{cor:cor_3}. In addition, even though the rate loss due to pure one-bit quantization is significant in the high SNR regime, the GMI in the low SNR regime closely approaches those of $K>0$, as predicted by Corollary \ref{cor:cor_2}.
\begin{figure*}
\centering
\subfigure[$\mathrm{SNR}=-10\ \mathrm{dB}$]{\label{fig:fig_11_a}
\includegraphics[width=0.32 \textwidth]{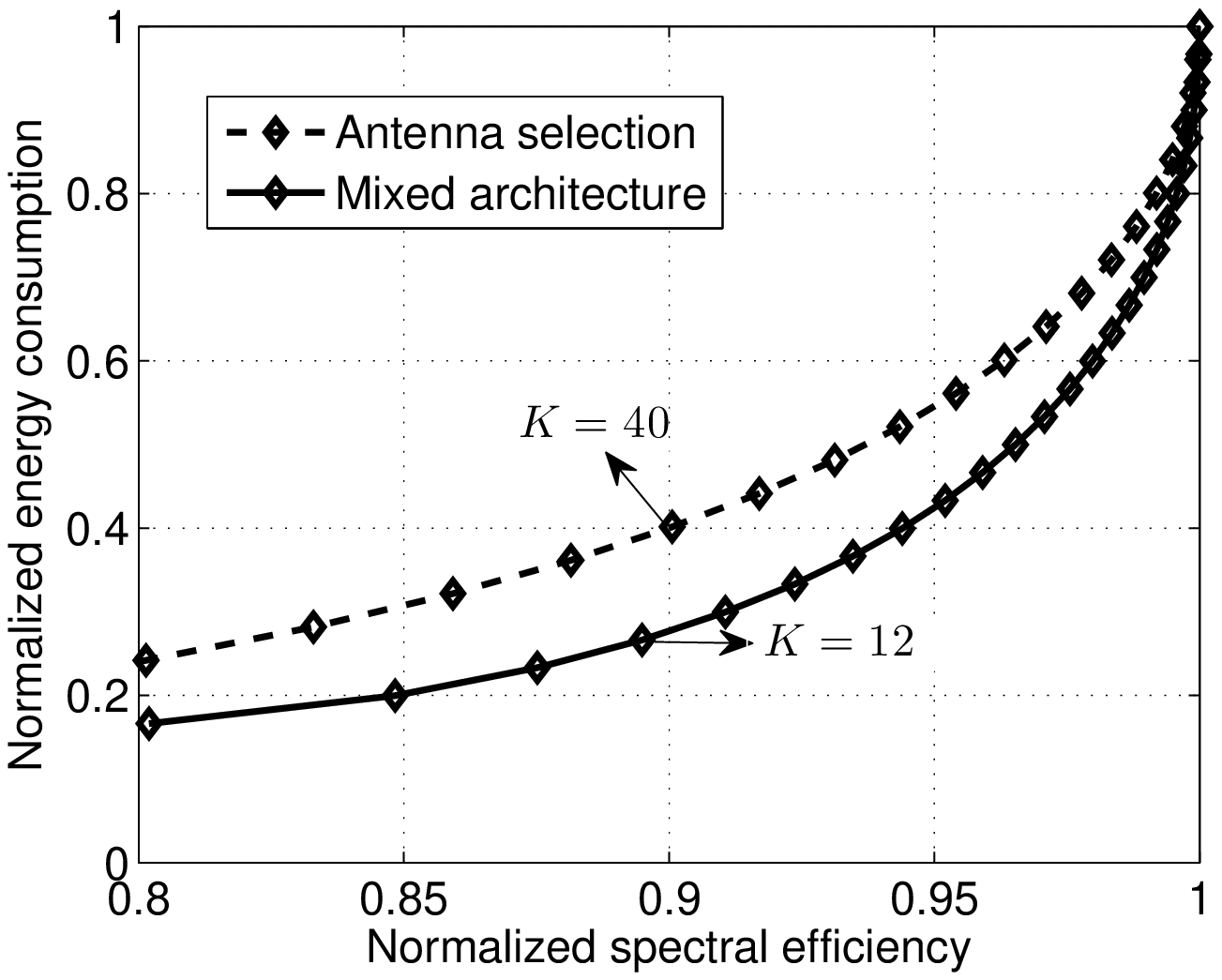}}
\subfigure[$\mathrm{SNR}=-5\ \mathrm{dB}$]{\label{fig:fig_11_b}
\includegraphics[width=0.32 \textwidth]{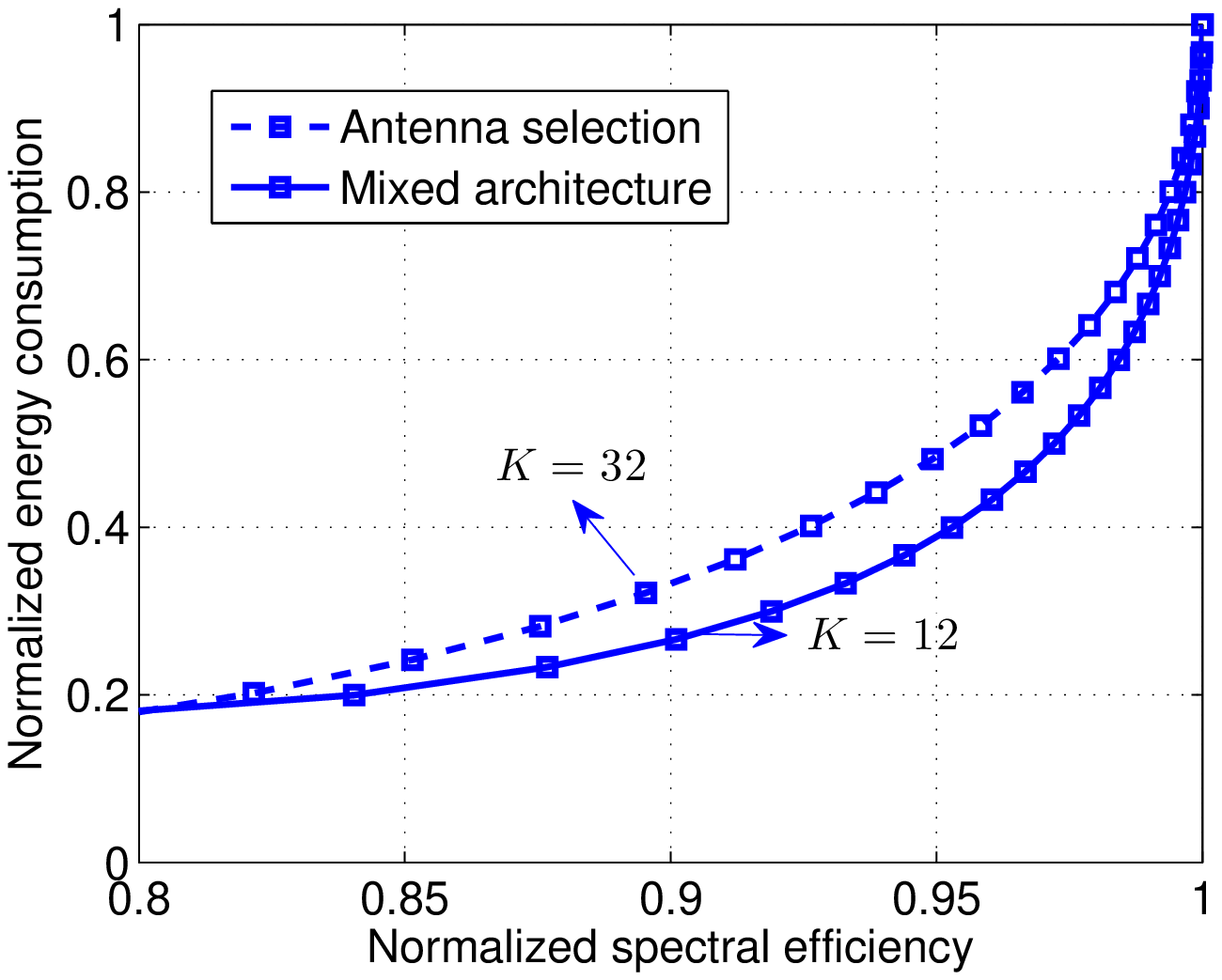}}
\subfigure[$\mathrm{SNR}=0\ \mathrm{dB}$]{\label{fig:fig_11_c}
\includegraphics[width=0.32 \textwidth]{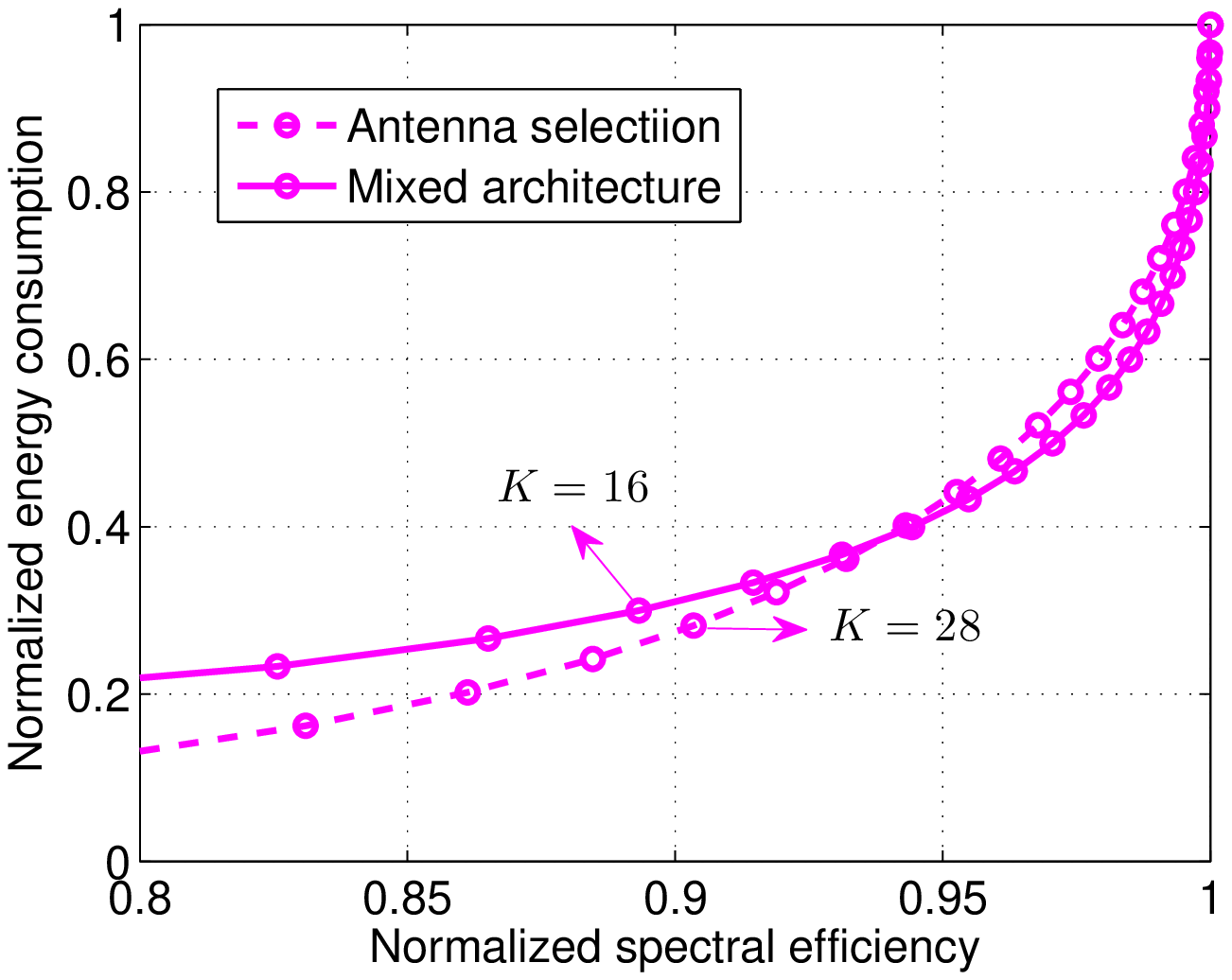}}
\caption{Energy efficiency comparison in the single-user scenario, $N=100$.}
\label{fig:fig_10}
\end{figure*}
\begin{figure*}
\centering
\subfigure[$\mathrm{SNR}=-5\ \mathrm{dB}$]{\label{fig:fig_12_a}
\includegraphics[width=0.32 \textwidth]{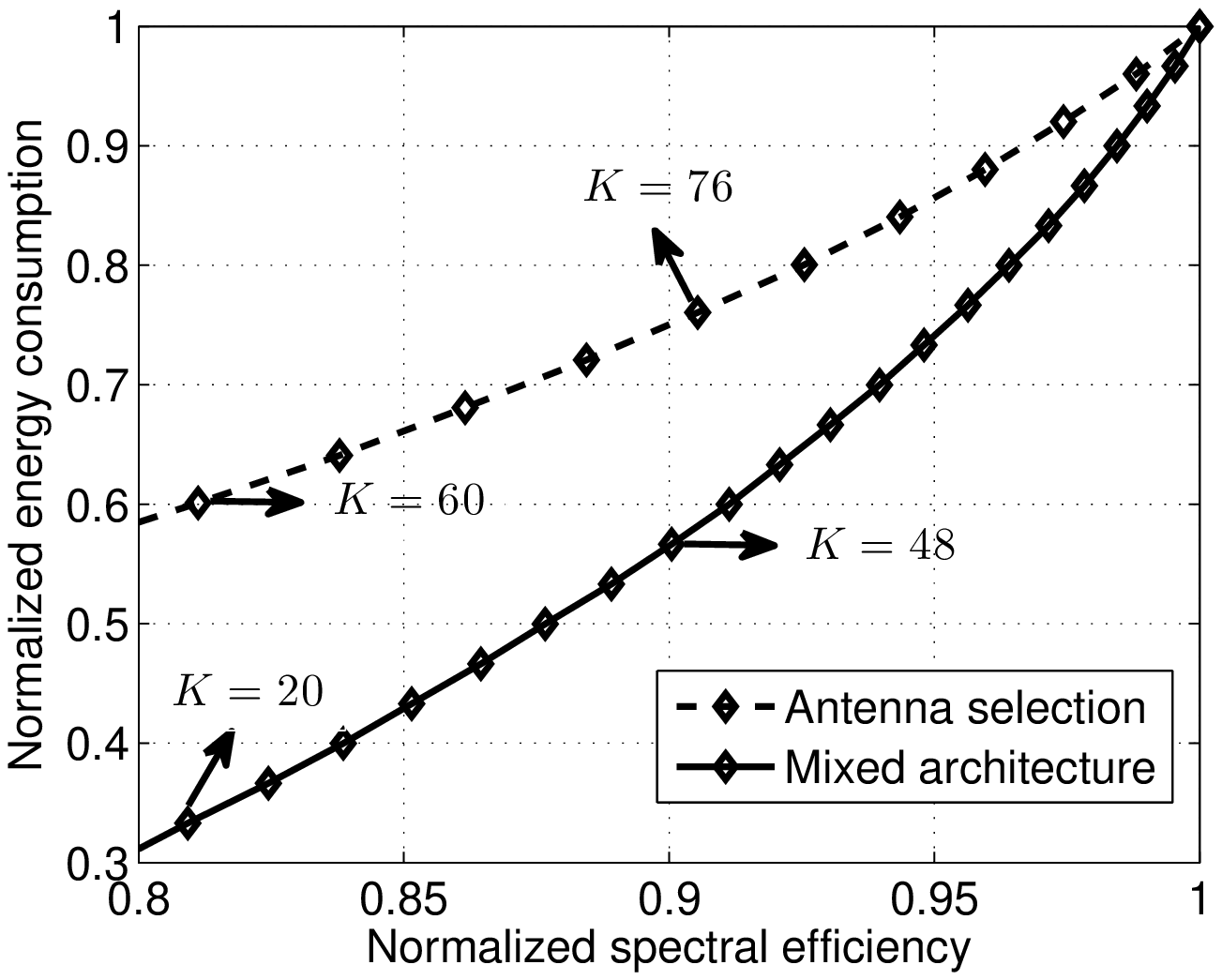}}
\subfigure[$\mathrm{SNR}=0\ \mathrm{dB}$]{\label{fig:fig_12_b}
\includegraphics[width=0.32 \textwidth]{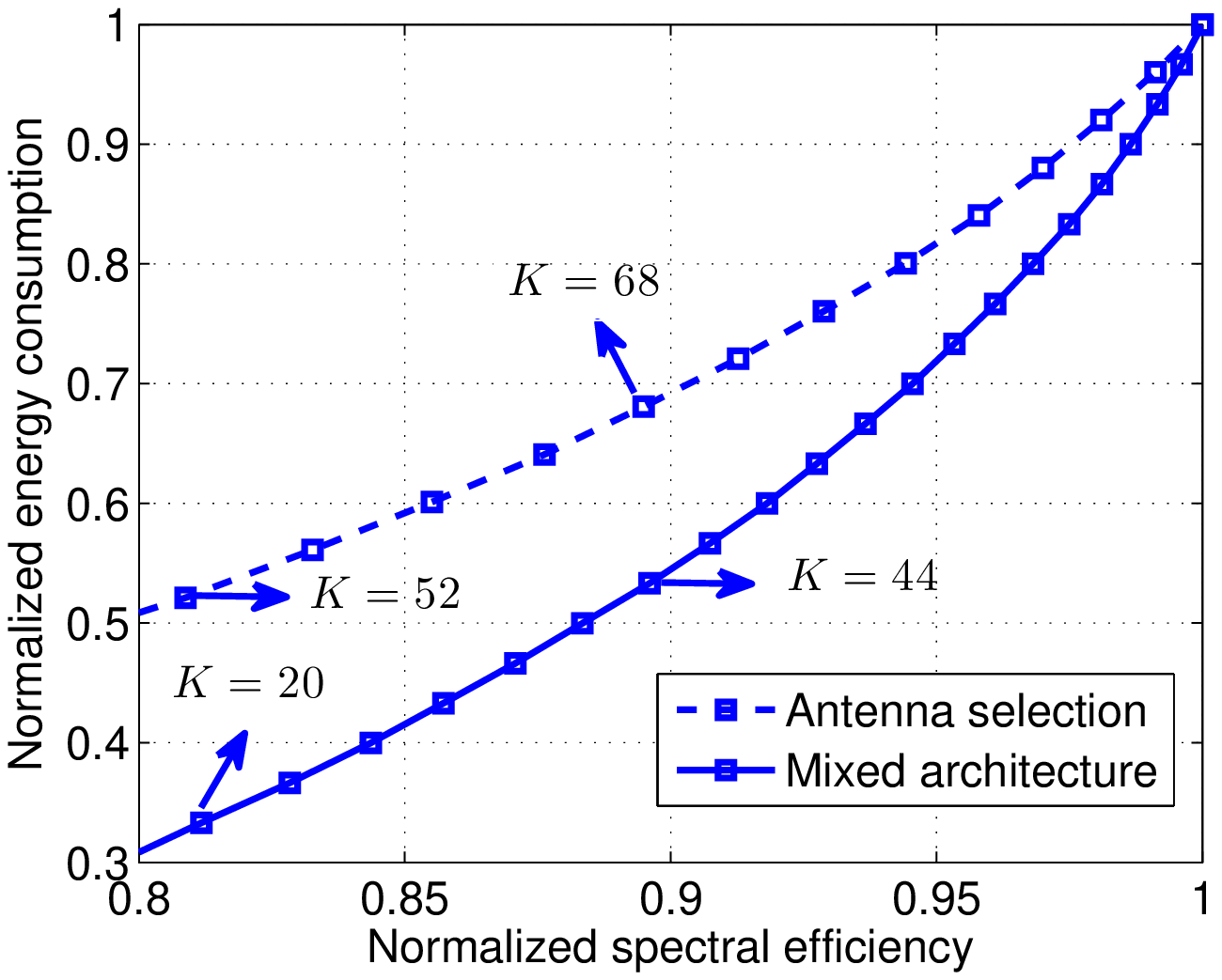}}
\subfigure[$\mathrm{SNR}=5\ \mathrm{dB}$]{\label{fig:fig_12_c}
\includegraphics[width=0.32 \textwidth]{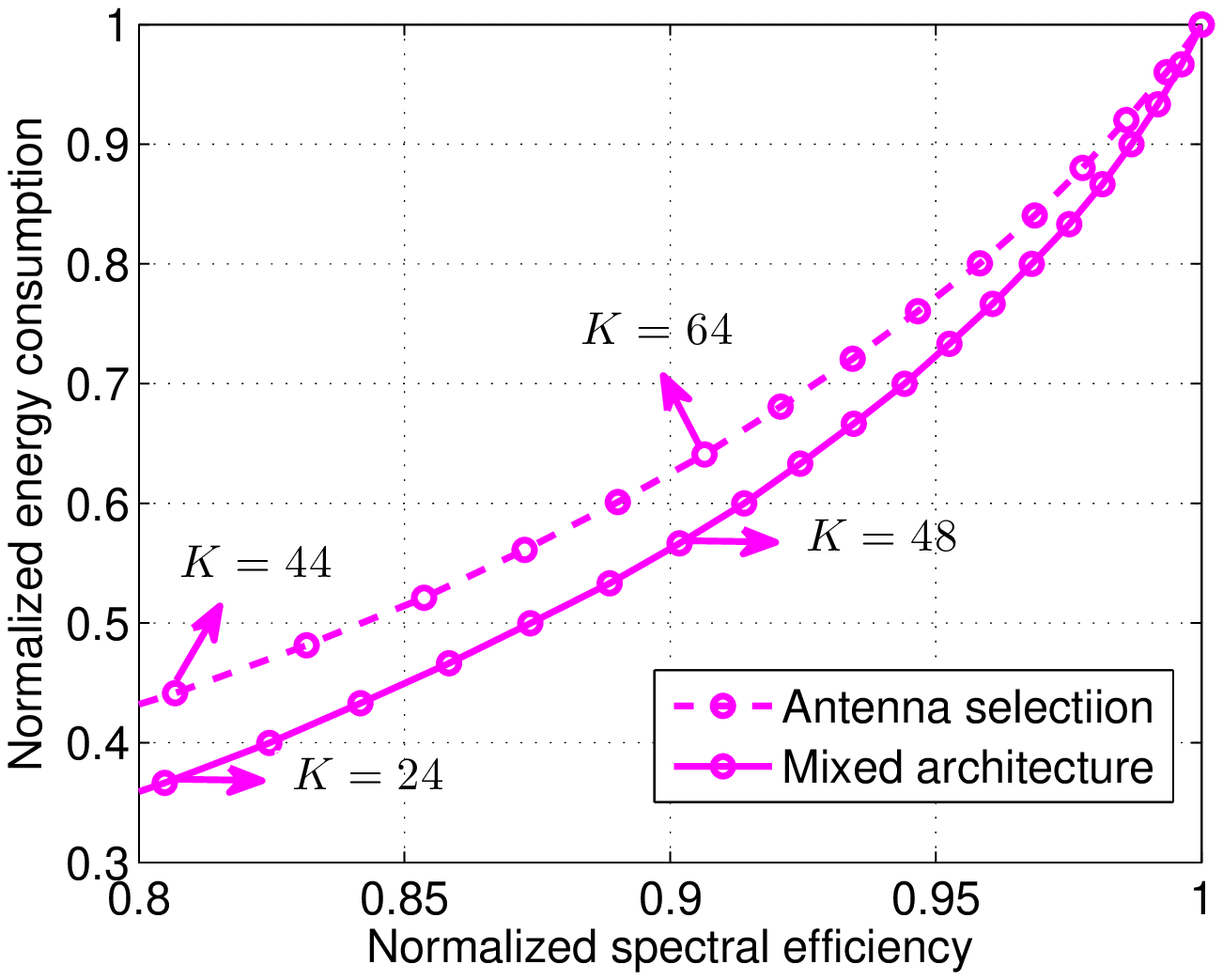}}
\caption{Energy efficiency comparison in the multi-user scenario, $N=100$, $M=10$.}
\label{fig:fig_11}
\end{figure*}

Then we examine the performance gain of Gaussian dithering. For given $N$ and $\mathrm{SNR}$, we optimize the threshold $\mathcal{T}$ assuming $K=0$, and then take the resulting $\mathcal{T}_{\mathrm{opt}}$ to evaluate the performance gain with $K>0$. Figure \ref{fig:fig_7} indicates that Gaussian dithering are able to achieve promising improvement in the spectral efficiency, especially for the case of $K=0$. Increasing either $K$ or $\mathrm{SNR}$, however, the benefit of dithering for $K>0$ decays gradually, since the contribution of high-resolution ADCs tends to be dominating.

\subsection{GMI for Ergodic Fading MU-MIMO Channel}
Now, we examine the feasibility of the mixed-ADC architecture in the multi-user scenario. The performance comparison between random and norm-based ADC switch schemes is given by Figure \ref{fig:fig_8}. We notice that though the norm-based ADC switch is only analytically validated in low SNR regime, it does achieve better performance. Moreover, the lower and upper bounds of the GMI for each scheme still virtually coincide with each other.

Figure \ref{fig:fig_9} compares the achievable spectral efficiency of the mixed-ADC architecture with that of conventional architecture and antenna selection (using linear MMSE receiver for a fair comparison). Similar to the conclusion we obtained for the single-user scenario, here the mixed-ADC architecture with a small number of high-resolution ADCs also attains a large fraction of the rate of conventional architecture. As a numerical evidence, when $\mathrm{SNR}=0$ dB and $N=100$, norm-based ADC switch with $K=10$ achieves 77\% of the per-user rate of conventional architecture, and this number rises to 81\% when we have $K=20$. Meanwhile, the mixed-ADC architecture also achieves a noticeably higher spectral efficiency than antenna selection.

\subsection{Energy Efficiency}
We evaluate the energy efficiency improvement of the mixed-ADC architecture as well as antenna selection, taking conventional architecture as a baseline. We emphasize that spectral efficiency should never be excessively sacrificed for energy efficiency, thus confining the normalized spectral efficiency to 80\% - 100\%.

Figure \ref{fig:fig_10} illustrates the numerical results for a single-user system. We notice that, if 10\% spectral efficiency degradation is allowed, then antenna selection can achieve more than 60\% energy consumption reduction, and beyond that, the mixed-ADC architecture can further reduce the energy consumption by about 10\%, in low to moderate SNR regime. Besides, it is perhaps worth noting that, in the high SNR regime, antenna selection may achieve higher energy efficiency than the mixed-ADC architecture, since now one-bit ADCs are getting less beneficial as demonstrated by Corollary \ref{cor:cor_3}.

Regarding the multi-user scenario, Figure \ref{fig:fig_11} reveals more pronounced superiority of the mixed-ADC architecture over antenna selection. The mixed-ADC architecture always outperforms antenna selection throughout the considered SNR range, and we note that the gap will further increase as the system load (i.e., the number of users $M$) increases. For the system parameters in Figure \ref{fig:fig_11}, it appears that spectral efficiency and energy efficiency arrive at an attractive tradeoff at $K\approx 20$, where we sacrifice a 20\% loss on spectral efficiency to trade for a 70\% reduction on energy consumption.

\section{Conclusion}
\label{sect:conclusion}
The numerous BS antennas enable massive MIMO systems to achieve unprecedented gains in both spectral efficiency and radiated energy efficiency, but also make the hardware cost and circuit power consumption increase unbearably, demanding energy-efficient design of transceivers. In this paper, we propose a mixed-ADC receiver architecture for the uplink, and leverage GMI to analytically evaluate its achievable data rates under various scenarios. Numerical results demonstrate that the mixed-ADC architecture with a relatively small number of high-resolution ADCs is able to achieve a large fraction of the channel capacity of conventional architecture, while reduce the energy consumption considerably even compared with antenna selection, for both single-user and multi-user scenarios. We envision the mixed-ADC architecture as a compelling choice for energy-efficient massive MIMO systems.

A number of interesting and important problems remain unsolved beyond this paper, such as designing the optimal ADC switch scheme for any SNR, especially for the multi-user scenario; making full use of the available one-bit ADCs when acquiring the CSI; extending the analysis to hardware impairment models besides ADC; among others. Additionally, in order to make this approach effective for wideband channels which are more prevailing in the future communication systems, it is particularly crucial to extend the analysis to frequency-selective fading channels. When one adopts multi-carrier transceiver architectures like OFDM, since one-bit ADCs are applied in the time domain rather than the frequency domain, severe inter-carrier interference due to quantization is inevitable and thus the decoder needs to properly account for this, say, by using a vectorized nearest-neighbor decoding algorithm and evaluating the resulting GMI. This is feasible but beyond the scope of this paper, and is currently treated in a separate work.

\section*{Appendix}
\subsection{Derivation of $\kappa(\mathbf{w},\bm{\delta})$}
\label{app:app_1}
We first introduce two lemmas that will help us derive a closed-form expression of $\kappa(\mathbf{w},\bm{\delta})$.
\begin{lem}
\label{lem:lem_1}
For zero-mean real Gaussian random variables $S$ and $T$ with covariance matrix $\mathbf{K}$, letting $\phi(s,t)$ denote their joint probability density function (PDF) and $\rho$ represent their correlation coefficient, we have
\begin{equation}
\mathbb{E}[\mathrm{sgn}(S)\cdot\mathrm{sgn}(T)]=\frac{2}{\pi}\mathrm{arcsin}(\rho).
\label{equ:equ_12}
\end{equation}
\end{lem}

\begin{proof}
Applying \cite[Prop. 2]{koch2010increased}, we obtain the following relationship,
\begin{equation}
\int_{0}^{\infty}\int_{0}^{\infty}\phi(s,t)\mathrm{d}t\mathrm{d}s=
\frac{1}{4}+\frac{1}{2\pi}\mathrm{arcsin}(\rho).
\label{equ:equ_13}
\end{equation}
Then exploiting the symmetry of $\phi(s,t)$, it is straightforward to verify that
\begin{eqnarray}
&&\mathbb{E}[\mathrm{sgn}(S)\cdot\mathrm{sgn}(T)]\nonumber\\
&=&\iint_{st>0}  \phi(s,t)\mathrm{d}t\mathrm{d}s-\iint_{st<0}\phi(s,t)\mathrm{d}t\mathrm{d}s\nonumber\\
&=&2\iint_{st>0}\phi(s,t)\mathrm{d}t\mathrm{d}s-1\nonumber\\
&=&4\int_{0}^{\infty}\int_{0}^{\infty}\phi(s,t)\mathrm{d}t\mathrm{d}s-1\nonumber\\
&=&\frac{2}{\pi}\mathrm{arcsin}(\rho).
\label{equ:equ_14}
\end{eqnarray}
\end{proof}

\begin{lem}
\label{lem:lem_2}
For independent complex Gaussian random variables $S\sim\mathcal{CN}(0,\sigma_s^2)$ and $T\sim\mathcal{CN}(0,\sigma_t^2)$, we have
\begin{eqnarray}
\mathbb{E}[S^*\cdot\mathrm{sgn}(S+T)]&=&\mathbb{E}[S\cdot\mathrm{sgn}^*(S+T)]\nonumber\\
&=&\sigma_s^2\sqrt{\frac{4}{\pi(\sigma_s^2+\sigma_t^2)}}.
\label{equ:equ_15}
\end{eqnarray}
\end{lem}

\begin{proof}
With some manipulation, we have
\begin{eqnarray}
&&\mathbb{E}[S^*\cdot\mathrm{sgn}(S+T)]\nonumber\\
&=&\mathbb{E}[S^{\mathrm{R}}\cdot\mathrm{sgn}(S^{\mathrm{R}}+T^{\mathrm{R}})]
+\mathbb{E}[S^{\mathrm{I}}\cdot\mathrm{sgn}(S^{\mathrm{I}}+T^{\mathrm{I}})]+\nonumber\\
&&i\cdot\mathbb{E}[S^{\mathrm{R}}\cdot\mathrm{sgn}(S^{\mathrm{I}}+T^{\mathrm{I}})]
-i\cdot\mathbb{E}[S^{\mathrm{I}}\cdot\mathrm{sgn}(S^{\mathrm{R}}+T^{\mathrm{R}})]\nonumber\\
&\overset{(a)}{=}&\frac{\sigma_s^2}{2}\sqrt{\frac{2}{\pi(\sigma_s^2/2+\sigma_t^2/2)}}
+\frac{\sigma_s^2}{2}\sqrt{\frac{2}{\pi(\sigma_s^2/2+\sigma_t^2/2)}}\nonumber\\
&=&\sigma_s^2\sqrt{\frac{4}{\pi(\sigma_s^2+\sigma_t^2)}},
\label{equ:equ_16}
\end{eqnarray}
where (a) follows from \cite[Eq. (19)]{zhang2012general}, the independence between $S^{\mathrm{R}}$ and $S^{\mathrm{I}}+T^{\mathrm{I}}$, as well as between $S^{\mathrm{I}}$ and $S^{\mathrm{R}}+T^{\mathrm{R}}$.
\end{proof}

Now we are ready to evaluate $|\mathbb{E}[f^{*}(x,\mathbf{h},\mathbf{z})\cdot x]|^2$ and $\mathbb{E}[|f(x,\mathbf{h},\mathbf{z})|^2]$. For given $\mathbf{w}$ and $\bm{\delta}$, we have
\begin{equation}
|\mathbb{E}[f^{*}(x,\mathbf{h},\mathbf{z})\cdot x]|^2
=|\mathbf{w}^{T}\mathbf{R}_{\mathbf{r}x}^*|^2
=\mathbf{w}^{H}\mathbf{R}_{\mathbf{r}x}\mathbf{R}_{\mathbf{r}x}^{H}\mathbf{w},
\label{equ:app_2}
\end{equation}
where $\mathbf{R}_{\mathbf{r}x}\triangleq\mathbb{E}[\mathbf{r}x^*]$ is the correlation vector between $\mathbf{r}$ and $x$, whose $n$-th element is
\begin{eqnarray}
&&(\mathbf{R}_{\mathbf{r}x})_n\nonumber\\
&=&\delta_n\cdot\mathbb{E}[x^*\cdot(h_n x+z_n)]+\bar{\delta}_n\cdot\mathbb{E}[x^{*}\cdot\mathrm{sgn}(h_n x+z_n)]\nonumber\\
&\overset{(a)}{=}&\delta_n\cdot h_n\mathcal{E}_\mathrm{s}+\bar{\delta}_n\cdot h_n\mathcal{E}_\mathrm{s}\sqrt{\frac{4}{\pi(|h_n|^2\mathcal{E}_\mathrm{s}+1)}}\nonumber\\
&=&h_n\mathcal{E}_\mathrm{s}\left[\delta_n+\bar{\delta}_n\cdot\sqrt{\frac{4}{\pi(|h_n|^2\mathcal{E}_\mathrm{s}+1)}}\right].
\label{equ:app_3}
\end{eqnarray}
Here, (a) follows directly from Lemma \ref{lem:lem_2}.

On the other hand, it is straightforward that
\begin{equation}
\mathbb{E}[|f(x,\mathbf{h},\mathbf{z})|^2]
=\mathbb{E}[\mathbf{w}^{H}\mathbf{r}\mathbf{r}^{H}\mathbf{w}]
=\mathbf{w}^{H}\mathbf{R}_{\mathbf{r}\mathbf{r}}\mathbf{w},
\label{equ:app_4}
\end{equation}
where $\mathbf{R}_{\mathbf{r}\mathbf{r}}\triangleq\mathbb{E}[\mathbf{r}\mathbf{r}^{H}]$ is the covariance matrix of $\mathbf{r}$. The diagonal elements of $\mathbf{R}_{\mathbf{rr}}$ are given by
\begin{eqnarray}
&&(\mathbf{R}_{\mathbf{rr}})_{n,n}\nonumber\\
&\!\!\!\!=\!\!\!\!&\mathbb{E}[|\delta_n\cdot(h_n x+z_n)+\bar{\delta}_n\cdot\mathrm{sgn}(h_n x+z_n)|^2]\nonumber\\
&\!\!\!\!=\!\!\!\!&\delta_n\cdot\mathbb{E}[|h_n x+z_n|^2]+\bar{\delta}_n\cdot \mathbb{E}[|\mathrm{sgn}(h_n x+z_n)|^2]\nonumber\\
&\!\!\!\!=\!\!\!\!&\delta_n\cdot(|h_n|^2\mathcal{E}_\mathrm{s}+1)+\bar{\delta}_n\cdot 2\nonumber\\
&\!\!\!\!=\!\!\!\!&1+\delta_n\cdot|h_n|^2\mathcal{E}_\mathrm{s}+\bar{\delta}_n,
\label{equ:app_5}
\end{eqnarray}
while the nondiagonal elements can be obtained by applying both Lemma \ref{lem:lem_1} and Lemma \ref{lem:lem_2}, as follows. First, applying Lemma \ref{lem:lem_2} we have
\begin{eqnarray}
&&\mathbb{E}[y_n\cdot\mathrm{sgn}^*(y_m)]\nonumber\\
&\!\!\!\!=\!\!\!\!&\mathbb{E}[(h_n x+z_n)\cdot\mathrm{sgn}^*(h_m x+z_m)]\nonumber\\
&\!\!\!\!=\!\!\!\!&\mathbb{E}\big[\mathbb{E}[(h_n x+z_n)\cdot\mathrm{sgn}^*(h_m x+z_m)|x,z_m]\big]\nonumber\\
&\!\!\!\!=\!\!\!\!&\mathbb{E}[h_n x\cdot\mathrm{sgn}^*(h_m x+z_m)]\nonumber\\
&\!\!\!\!=\!\!\!\!&h_n h_m^* \mathcal{E}_\mathrm{s}\sqrt{\frac{4}{\pi(|h_n|^2\mathcal{E}_\mathrm{s}+1)}},
\label{equ:app_6}
\end{eqnarray}
and analogously
\begin{equation}
\mathbb{E}[\mathrm{sgn}(y_n)\cdot y_m^*]=h_n h_m^* \mathcal{E}_\mathrm{s}\sqrt{\frac{4}{\pi(|h_m|^2\mathcal{E}_\mathrm{s}+1)}}.
\label{equ:app_7}
\end{equation}
Then, we turn to evaluate $\mathbb{E}[\mathrm{sgn}(y_n)\cdot\mathrm{sgn}^*(y_m)]$; that is
\begin{eqnarray}
&&\mathbb{E}[\mathrm{sgn}(y_n)\cdot\mathrm{sgn}^*(y_m)]\nonumber\\
&\!\!\!\!=\!\!\!\!&\mathbb{E}[\mathrm{sgn}(y_n^{\mathrm{R}})\cdot\mathrm{sgn}(y_m^{\mathrm{R}})]
+\mathbb{E}[\mathrm{sgn}(y_n^{\mathrm{I}})\cdot\mathrm{sgn}(y_m^{\mathrm{I}})]-\nonumber\\
&&i\cdot\mathbb{E}[\mathrm{sgn}(y_n^{\mathrm{R}})\cdot\mathrm{sgn}(y_m^{\mathrm{I}})]
+i\cdot\mathbb{E}[\mathrm{sgn}(y_n^{\mathrm{I}})\cdot\mathrm{sgn}(y_m^{\mathrm{R}})]\nonumber\\
&\!\!\!\!=\!\!\!\!&\frac{2}{\pi}\mathrm{arcsin}(\rho_{y_n^{\mathrm{R}},y_m^{\mathrm{R}}})
+\frac{2}{\pi}\mathrm{arcsin}(\rho_{y_n^{\mathrm{I}},y_m^{\mathrm{I}}})-\nonumber\\
&&i\cdot\frac{2}{\pi}\mathrm{arcsin}(\rho_{y_n^{\mathrm{R}},y_m^{\mathrm{I}}})
+i\cdot\frac{2}{\pi}\mathrm{arcsin}(\rho_{y_n^{\mathrm{I}},y_m^{\mathrm{R}}}),
\label{equ:app_8}
\end{eqnarray}
where the last equation follows from Lemma \ref{lem:lem_1}. To proceed, we need to evaluate some correlation coefficients, e.g., $\rho_{y_n^{\mathrm{R}},y_m^{\mathrm{R}}}$, which is given as
\begin{eqnarray}
&&\rho_{y_n^{\mathrm{R}},y_m^{\mathrm{R}}}\nonumber\\
&\!\!\!\!=\!\!\!\!&\frac{\mathbb{E}[y_n^{\mathrm{R}}y_m^{\mathrm{R}}]}{\sqrt{\mathbb{E}[(y_n^{\mathrm{R}})^2]}\sqrt{\mathbb{E}[(y_m^{\mathrm{R}})^2]}}\nonumber\\
&\!\!\!\!=\!\!\!\!&\frac{\mathbb{E}[(h_n^{\mathrm{R}} x^{\mathrm{R}}\!-\!h_n^{\mathrm{I}} x^{\mathrm{I}}\!+\!z_n^{\mathrm{R}})
(h_m^{\mathrm{R}} x^{\mathrm{R}}\!-\!h_m^{\mathrm{I}} x^{\mathrm{I}}\!+\!z_m^{\mathrm{R}})]}
{\sqrt{\mathbb{E}[(h_n^{\mathrm{R}} x^{\mathrm{R}}\!-\!h_n^{\mathrm{I}} x^{\mathrm{I}}\!+\!z_n^{\mathrm{R}})^2]}
\sqrt{\mathbb{E}[(h_m^{\mathrm{R}} x^{\mathrm{R}}\!-\!h_m^{\mathrm{I}} x^{\mathrm{I}}\!+\!z_m^{\mathrm{R}})^2]}}\nonumber\\
&\!\!\!\!=\!\!\!\!&\frac{(h_n^{\mathrm{R}}h_m^{\mathrm{R}}+h_n^{\mathrm{I}}h_m^{\mathrm{I}})\frac{\mathcal{E}_\mathrm{s}}{2}}
{\sqrt{|h_n|^2\frac{\mathcal{E}_\mathrm{s}}{2}+\frac{1}{2}}\sqrt{|h_m|^2\frac{\mathcal{E}_\mathrm{s}}{2}+\frac{1}{2}}}\nonumber\\
&\!\!\!\!=\!\!\!\!&\frac{(h_nh_m^*)^{\mathrm{R}}\mathcal{E}_\mathrm{s}}{\sqrt{|h_n|^2\mathcal{E}_\mathrm{s}+1}\sqrt{|h_m|^2\mathcal{E}_\mathrm{s}+1}}.
\label{equ:app_9}
\end{eqnarray}
Besides, following essentially the same line we have
\begin{eqnarray}
\rho_{y_n^{\mathrm{I}},y_m^{\mathrm{I}}}&\!\!\!\!=\!\!\!\!&\rho_{y_n^{\mathrm{R}},y_m^{\mathrm{R}}},\nonumber\\
\rho_{y_n^{\mathrm{I}},y_m^{\mathrm{R}}}&\!\!\!\!=\!\!\!\!&-\rho_{y_n^{\mathrm{R}},y_m^{\mathrm{I}}}
\!=\!\frac{(h_nh_m^*)^{\mathrm{I}}\mathcal{E}_\mathrm{s}}{\sqrt{|h_n|^2\mathcal{E}_\mathrm{s}+1}\sqrt{|h_m|^2\mathcal{E}_\mathrm{s}+1}}.
\label{equ:app_10}
\end{eqnarray}
Now we can combine \eqref{equ:app_8}-\eqref{equ:app_10} to get $\mathbb{E}[\mathrm{sgn}(y_n)\cdot\mathrm{sgn}^*(y_m)]$ as follows
\begin{eqnarray}
&&\mathbb{E}[\mathrm{sgn}(y_n)\cdot\mathrm{sgn}^*(y_m)]\nonumber\\
&\!\!\!\!=\!\!\!\!&\frac{4}{\pi}\mathrm{arcsin}\left(\frac{(h_n h_m^*)^{\mathrm{R}}\mathcal{E}_\mathrm{s}}
{\sqrt{|h_n|^2\mathcal{E}_\mathrm{s}+1}\sqrt{|h_m|^2\mathcal{E}_\mathrm{s}+1}}\right)+\nonumber\\
&&i\cdot\frac{4}{\pi}\mathrm{arcsin}\left(\frac{(h_n h_m^*)^{\mathrm{I}}\mathcal{E}_\mathrm{s}}
{\sqrt{|h_n|^2\mathcal{E}_\mathrm{s}+1}\sqrt{|h_m|^2\mathcal{E}_\mathrm{s}+1}}\right),
\label{equ:app_11}
\end{eqnarray}
Further, from \eqref{equ:app_6}, \eqref{equ:app_7} and \eqref{equ:app_11}, we obtain $(\mathbf{R}_{rr})_{n,m}$, given as
\begin{eqnarray}
&&(\mathbf{R}_{\mathbf{rr}})_{n,m}\nonumber\\
&\!\!\!\!=\!\!\!\!&\delta_n\delta_m\!\cdot\!\mathbb{E}[y_ny_m^*]
+\delta_n\bar{\delta}_m\!\cdot\!\mathbb{E}[y_n\!\cdot\!\mathrm{sgn}^*(y_m)]+\nonumber\\
&&\bar{\delta}_n\delta_m\!\cdot\!\mathbb{E}[\mathrm{sgn}(y_n)\!\cdot\! y_m^*]
+\bar{\delta}_n\bar{\delta}_m\!\cdot\!\mathbb{E}[\mathrm{sgn}(y_n)\!\cdot\!\mathrm{sgn}^*(y_m)]\nonumber\\
&\!\!\!\!=\!\!\!\!&h_n h_m^* \mathcal{E}_\mathrm{s}\Bigg[\delta_n\delta_m+\delta_n\bar{\delta}_m\cdot\sqrt{\frac{4}{\pi(|h_m|^2\mathcal{E}_\mathrm{s}+1)}}+\nonumber\\
&&\ \ \ \ \ \ \ \ \ \ \ \bar{\delta}_n\delta_m\cdot\sqrt{\frac{4}{\pi(|h_n|^2\mathcal{E}_\mathrm{s}+1)}}\Bigg]+\nonumber\\
&&\bar{\delta}_n\bar{\delta}_m\!\cdot\!\frac{4}{\pi}\Bigg[
\mathrm{arcsin}\Bigg(\frac{(h_nh_m^*)^{\mathrm{R}}\mathcal{E}_\mathrm{s}}{\sqrt{|h_n|^2\mathcal{E}_\mathrm{s}+1}\sqrt{|h_m|^2\mathcal{E}_\mathrm{s}+1}}\Bigg)+\nonumber\\ &&\ \ \ \ \ \ \ \ \ \ \ i\!\cdot\!\mathrm{arcsin}\Bigg(\frac{(h_nh_m^*)^{\mathrm{I}}\mathcal{E}_\mathrm{s}}{\sqrt{|h_n|^2\mathcal{E}_\mathrm{s}+1}\sqrt{|h_m|^2\mathcal{E}_\mathrm{s}+1}}\Bigg)
\Bigg].\nonumber\\
\label{equ:app_12}
\end{eqnarray}
Thus we conclude the proof.
\subsection{Asymptotic behavior of $I_{\mathrm{GMI}}(\mathbf{w}_{\mathrm{opt}},\bm{\delta})$ in low SNR regime}
\label{app:app_2}
For simplicity of exposition, we define
\begin{equation}
\mathbf{R}_{\mathbf{r}x}^{0}\triangleq\lim_{\mathcal{E}_\mathrm{s}\rightarrow 0}\frac{1}{\mathcal{E}_\mathrm{s}}\mathbf{R}_{\mathbf{r}x},\ \
\mathbf{R}_{\mathbf{rr}}^{0}\triangleq\lim_{\mathcal{E}_\mathrm{s}\rightarrow 0}\mathbf{R}_{\mathbf{rr}}.
\label{equ:app_13}
\end{equation}
Then from \eqref{equ:equ_18} and \eqref{equ:equ_19}, it is straightforward to verify that
\begin{eqnarray}
(\mathbf{R}_{\mathbf{r}x}^{0})_n&=&h_n\left[\delta_n+\bar{\delta}_n\cdot\frac{2}{\sqrt{\pi}}\right],\nonumber\\
\mathbf{R}_{\mathbf{rr}}^{0}&=&\mathrm{diag}(1+\bar{\delta}_1,...,1+\bar{\delta}_n,...,1+\bar{\delta}_N).
\label{equ:app_14}
\end{eqnarray}
Thereby we examine the asymptotic behavior of $\kappa(\mathbf{w}_{\mathrm{opt}},\bm{\delta})$ as $\mathcal{E}_\mathrm{s}\rightarrow 0$; that is
\begin{eqnarray}
&&\lim_{\mathcal{E}_\mathrm{s}\rightarrow 0}\frac{\kappa(\mathbf{w}_{\mathrm{opt}},\bm{\delta})}{\mathcal{E}_\mathrm{s}}\nonumber\\
&\overset{(a)}{=}&\lim_{\mathcal{E}_\mathrm{s}\rightarrow 0}\Big(\frac{1}{\mathcal{E}_\mathrm{s}}\mathbf{R}_{\mathbf{r}x}\Big)^{H}
\mathbf{R}_{\mathbf{rr}}^{-1}\Big(\frac{1}{\mathcal{E}_\mathrm{s}}\mathbf{R}_{\mathbf{r}x}\Big)\nonumber\\
&\overset{(b)}{=}&\Big(\lim_{\mathcal{E}_\mathrm{s}\rightarrow 0}\frac{1}{\mathcal{E}_\mathrm{s}}\mathbf{R}_{\mathbf{r}x}\Big)^{H}
\Big(\lim_{\mathcal{E}_\mathrm{s}\rightarrow 0}\mathbf{R}_{\mathbf{rr}}^{-1}\Big)
\Big(\lim_{\mathcal{E}_\mathrm{s}\rightarrow 0}\frac{1}{\mathcal{E}_\mathrm{s}}\mathbf{R}_{\mathbf{r}x}\Big)\nonumber\\
&\overset{(c)}{=}&\big(\mathbf{R}_{\mathbf{r}x}^{0}\big)^{H}\big(\mathbf{R}_{\mathbf{rr}}^{0}\big)^{-1}
\big(\mathbf{R}_{\mathbf{r}x}^{0}\big)\nonumber\\
&=&\sum_{n=1}^{N}\frac{\left(\delta_n+\bar{\delta}_n\cdot\frac{4}{\pi}\right)|h_n|^2}{1+\bar{\delta}_n}\nonumber\\
&=&\sum_{n=1}^{N}\left(\delta_n+\bar{\delta}_n\cdot\frac{2}{\pi}\right)|h_n|^2,
\label{equ:app_15}
\end{eqnarray}
where (a) follows from \eqref{equ:equ_22}, (b) is obtained by applying the algebraic limit theorem since the limits of $\mathbf{R}_{\mathbf{r}x}^{0}/\mathcal{E}_\mathrm{s}$ and $\mathbf{R}_{\mathbf{rr}}$ exist, while (c) comes from the fact that the inverse of a nonsingular matrix is a continuous function of the elements of the matrix, i.e., $\lim_{\mathcal{E}_\mathrm{s}\rightarrow 0}\mathbf{R}_{\mathbf{rr}}^{-1}=(\lim_{\mathcal{E}_\mathrm{s}\rightarrow 0}\mathbf{R}_{\mathbf{rr}})^{-1}$ \cite{stewart1969continuity}. As a result, when $\mathcal{E}_\mathrm{s}\rightarrow 0$ we have
\begin{equation}
\kappa(\mathbf{w}_{\mathrm{opt}},\bm{\delta})=\sum_{n=1}^{N}\left(\delta_n+\bar{\delta}_n\cdot\frac{2}{\pi}\right)|h_n|^2\mathcal{E}_\mathrm{s}+o(\mathcal{E}_\mathrm{s}).
\label{equ:app_16}
\end{equation}
Noting that $\log(1+x/(1-x))=x+o(x)$, as $x\rightarrow 0$, we immediately have \eqref{equ:equ_28}.
\subsection{Asymptotic behavior of $I_{\mathrm{GMI}}(\mathbf{w}_{\mathrm{opt}},\bm{\delta})$ in high SNR regime}
\label{app:app_3}
For simplicity of exposition, we rearrange $\mathbf{h}$ and stack the channel coefficients corresponding to the antennas equipped with high-resolution ADCs in the first $K$ positions of $\underline{\mathbf{h}}$. To proceed, we further define
\begin{eqnarray}
\mathbf{p}&\triangleq&[\underline{h}_1,...,\underline{h}_K]^{T},\nonumber\\
\mathbf{q}&\triangleq&\left[\underline{h}_{K+1}/|\underline{h}_{K+1}|,...,\underline{h}_{N}/|\underline{h}_{N}|\right]^{T}.
\label{equ:app_21}
\end{eqnarray}

When $\mathcal{E}_\mathrm{s}$ tends to infinity, we have
\begin{equation}
\underline{h}_n\mathcal{E}_\mathrm{s}\sqrt{\frac{4}{\pi(|\underline{h}_n|^2\mathcal{E}_\mathrm{s}+1)}}
=\left[\sqrt{4\mathcal{E}_\mathrm{s}/\pi}+O(1/\sqrt{\mathcal{E}_\mathrm{s}})\right]\cdot\frac{\underline{h}_n}{|\underline{h}_n|},
\label{equ:app_22}
\end{equation}
for $n=K+1,...,N$. As a result, we are allowed to denote the deduced $\underline{\mathbf{R}}_{\mathbf{r}x}$ as
\begin{equation}
\underline{\mathbf{R}}_{\mathbf{r}x}=
\left[
    \begin{array}{c}
        \mathcal{E}_\mathrm{s}\mathbf{p}\\
        (\sqrt{4\mathcal{E}_\mathrm{s}/\pi}+O(1/\sqrt{\mathcal{E}_\mathrm{s}}))\mathbf{q}
    \end{array}
\right].
\label{equ:app_23}
\end{equation}
Besides, we denote by partitioned matrices $\underline{\mathbf{R}}_{\mathbf{rr}}$ and its inverse $\underline{\mathbf{R}}_{\mathbf{rr}}^{-1}$, i.e.,
\begin{eqnarray}
&&\underline{\mathbf{R}}_{\mathbf{rr}}\triangleq
\left[
    \begin{array}{cc}
        \mathbf{A} & \mathbf{U}\\
        \mathbf{U}^{H} & \mathbf{B}
    \end{array}
\right],\ \
\underline{\mathbf{R}}_{\mathbf{rr}}^{-1}\triangleq
\left[
    \begin{array}{cc}
        \mathbf{C} & \mathbf{V}\\
        \mathbf{V}^{H} & \mathbf{D}
    \end{array}
\right],
\label{equ:app_24}
\end{eqnarray}
in which the invertible square matrices $\mathbf{A}\in\mathbb{C}^{K\times K}$, $\mathbf{B}\in\mathbb{C}^{(N-K)\times (N-K)}$ and the rectangle matrix $\mathbf{U}\in\mathbb{C}^{K\times (N-K)}$ are taken to be
\begin{eqnarray}
\mathbf{A}&\!\!\!\!=\!\!\!\!&\mathbf{I}+\mathcal{E}_\mathrm{s}\mathbf{p}\mathbf{p}^{H},\nonumber\\
\mathbf{U}&\!\!\!\!=\!\!\!\!&(\sqrt{4\mathcal{E}_\mathrm{s}/\pi}+O(1/\sqrt{\mathcal{E}_\mathrm{s}}))\mathbf{p}\mathbf{q}^{H},\nonumber\\
(\mathbf{B})_{n,m}&\!\!\!\!=\!\!\!\!&\frac{4}{\pi}\Bigg[\arcsin\left(\frac{(\underline{h}_{n+K}\underline{h}_{m+K}^*)^{\mathrm{R}}}
{|\underline{h}_{n+K}\underline{h}_{m+K}^*|}\right)+\nonumber\\
&&i\cdot\arcsin\left(\frac{(\underline{h}_{n+K}\underline{h}_{m+K}^*)^{\mathrm{I}}}
{|\underline{h}_{n+K}\underline{h}_{m+K}^*|}\right)\Bigg]+O(1/\mathcal{E}_\mathrm{s}).\nonumber\\
\label{equ:app_25}
\end{eqnarray}
Then, applying the Sherman-Morrison formula \cite{horn2012matrix} and the inverse of partitioned matrix \cite{hotelling1943some}, we obtain
\begin{eqnarray}
\mathbf{C}&\!\!\!\!=\!\!\!\!&(\mathbf{A}-\mathbf{U}\mathbf{B}^{-1}\mathbf{U}^{H})^{-1}\nonumber\\
&\!\!\!\!=\!\!\!\!&\mathbf{I}-\frac{\mathcal{E}_\mathrm{s}\left[\pi-(4+O(1/\mathcal{E}_\mathrm{s}))\mathbf{q}^{H}\mathbf{B}^{-1}\mathbf{q}\right]\cdot\mathbf{p}\mathbf{p}^{H}}
{\pi+\mathcal{E}_\mathrm{s}\left[\pi-(4+O(1/\mathcal{E}_\mathrm{s}))\mathbf{q}^{H}\mathbf{B}^{-1}\mathbf{q}\right]\cdot\|\mathbf{p}\|^2},\nonumber\\
\mathbf{V}&\!\!\!\!=\!\!\!\!&-\mathbf{A}^{-1}\mathbf{U}(\mathbf{B}-\mathbf{U}^{H}\mathbf{A}^{-1}\mathbf{U})^{-1}\nonumber\\
&\!\!\!\!=\!\!\!\!&-\frac{\left[\sqrt{4\pi\mathcal{E}_\mathrm{s}}+O(1/\sqrt{\mathcal{E}_\mathrm{s}})\right]\cdot\mathbf{p}\mathbf{q}^{H}\mathbf{B}^{-1}}
{\pi+\mathcal{E}_\mathrm{s}\left[\pi-(4+O(1/\mathcal{E}_\mathrm{s}))\mathbf{q}^{H}\mathbf{B}^{-1}\mathbf{q}\right]\cdot\|\mathbf{p}\|^2},\nonumber\\
\mathbf{D}&\!\!\!\!=\!\!\!\!&(\mathbf{B}-\mathbf{U}^{H}\mathbf{A}^{-1}\mathbf{U})^{-1}\nonumber\\
&\!\!\!\!=\!\!\!\!&\mathbf{B}^{-1}
+\frac{\left[4\mathcal{E}_\mathrm{s}+O(1)\right]\|\mathbf{p}\|^2\cdot\mathbf{B}^{-1}\mathbf{q}\mathbf{q}^{H}\mathbf{B}^{-1}}
{\pi+\mathcal{E}_\mathrm{s}\left[\pi-(4+O(1/\mathcal{E}_\mathrm{s}))\mathbf{q}^{H}\mathbf{B}^{-1}\mathbf{q}\right]\cdot\|\mathbf{p}\|^2}.\nonumber\\
\label{equ:app_26}
\end{eqnarray}
With all of these, we are ready to simplify $\kappa(\mathbf{w}_{\mathrm{opt}},\bm{\delta})$; that is,
\begin{eqnarray}
&&\kappa(\mathbf{w}_{\mathrm{opt}},\bm{\delta})\nonumber\\
&=&\frac{1}{\mathcal{E}_\mathrm{s}}
\underline{\mathbf{R}}_{\mathbf{r}x}^{H}\underline{\mathbf{R}}_{\mathbf{rr}}^{-1}\underline{\mathbf{R}}_{\mathbf{r}x}\nonumber\\
&=&\mathcal{E}_\mathrm{s}\mathbf{p}^{H}\mathbf{C}\mathbf{p}
+2\left(\sqrt{4\mathcal{E}_\mathrm{s}/\pi}+O(1/\sqrt{\mathcal{E}_\mathrm{s}})\right)\mathbf{p}^{H}\mathbf{V}\mathbf{q}+\nonumber\\
&&\left(4/\pi+O(1/\mathcal{E}_\mathrm{s})\right)\mathbf{q}^{H}\mathbf{D}\mathbf{q}\nonumber\\
&=&\frac{\mathcal{E}_\mathrm{s}\left[\pi-(4+O(1/\mathcal{E}_\mathrm{s}))\mathbf{q}^{H}\mathbf{B}^{-1}\mathbf{q}\right]\cdot\|\mathbf{p}\|^2}
{\pi+\mathcal{E}_\mathrm{s}\left[\pi-(4+O(1/\mathcal{E}_\mathrm{s}))\mathbf{q}^{H}\mathbf{B}^{-1}\mathbf{q}\right]\cdot\|\mathbf{p}\|^2}+\nonumber\\
&&\frac{\left[4+O(1/\mathcal{E}_\mathrm{s})\right]\mathbf{q}^{H}\mathbf{B}^{-1}\mathbf{q}}
{\pi+\mathcal{E}_\mathrm{s}\left[\pi-(4+O(1/\mathcal{E}_\mathrm{s}))\mathbf{q}^{H}\mathbf{B}^{-1}\mathbf{q}\right]\cdot\|\mathbf{p}\|^2}.
\label{equ:app_27}
\end{eqnarray}
Finally, we get the effective SNR as
\begin{equation}
\frac{\kappa(\mathbf{w}_{\mathrm{opt}},\bm{\delta})}{1-\kappa(\mathbf{w}_{\mathrm{opt}},\bm{\delta})}=\|\mathbf{p}\|^2\mathcal{E}_\mathrm{s}
+\frac{[4+O(1/\mathcal{E}_\mathrm{s})]\mathbf{q}^{H}\mathbf{B}^{-1}\mathbf{q}}{\pi-[4+O(1/\mathcal{E}_\mathrm{s})]\mathbf{q}^{H}\mathbf{B}^{-1}\mathbf{q}}.
\label{equ:app_28}
\end{equation}

\section*{Acknowledgement}
The authors thank the Editor and Reviewers for valuable suggestions that help improve the structure and exposition of the paper.


\begin{IEEEbiography}[{\includegraphics[width=1in,height=1.25in,clip,keepaspectratio]{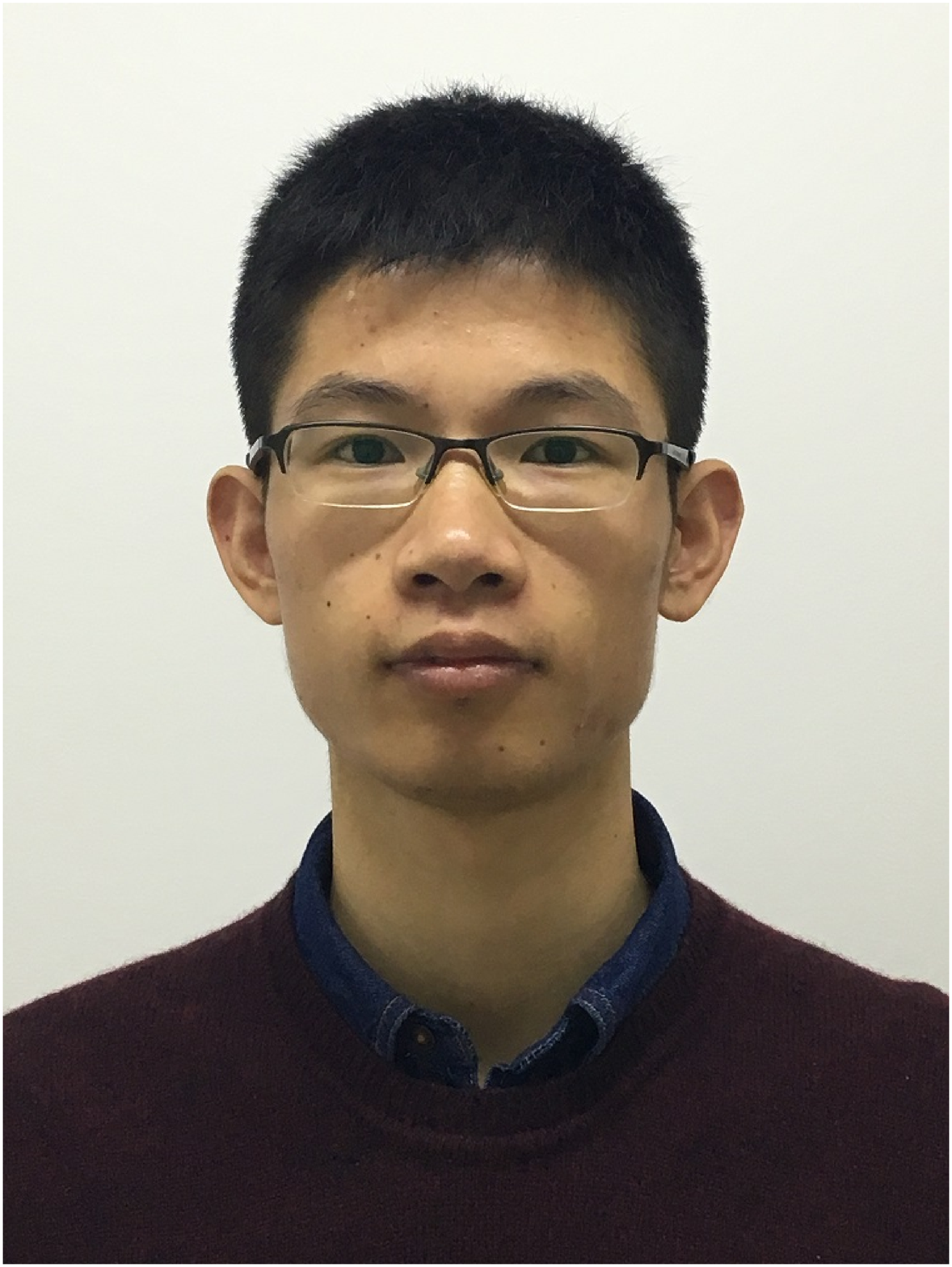}}]{Ning Liang}
received his B.E. degree in Communication Engineering from University of Science and Technology of China (USTC) in 2012. He is now a Ph.D. student in Wireless Communications at USTC, Hefei, China. His research interests include network interference analysis and low-complexity receiver design for massive MIMO systems.
\end{IEEEbiography}
\begin{IEEEbiography}[{\includegraphics[width=1in,height=1.25in,clip,keepaspectratio]{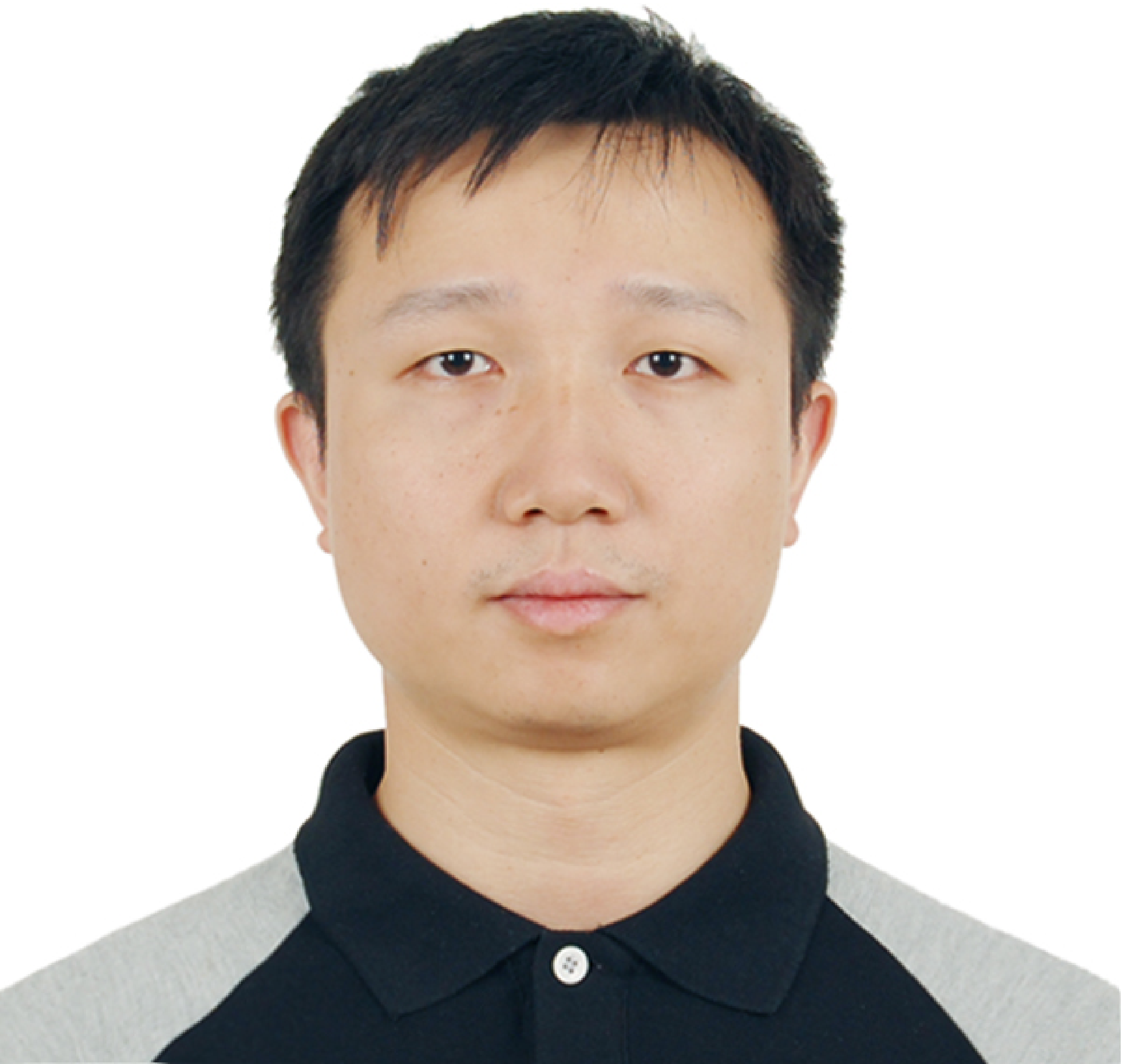}}]{Wenyi Zhang}
(S-00, M-07, SM-11) is with the faculty of Department of Electronic Engineering and Information Science, University of Science and Technology of China. Prior to that, he was affiliated with the Communication Science Institute, University of Southern California, as a postdoctoral research associate, and with Qualcomm Incorporated, Corporate Research and Development. He studied in Tsinghua University and obtained his Bachelor's degree in Automation in 2001; he studied in the University of Notre Dame, Indiana, USA, and obtained his Master's and Ph.D. degrees, both in Electrical Engineering, in 2003 and 2006, respectively. His research interests include wireless communications and networking, information theory, and statistical signal processing.
\end{IEEEbiography}
\end{document}